\documentclass[submission]{eptcs}
\pdfoutput=1

\usepackage{amsmath} 
\usepackage{amsthm} 
\usepackage{amssymb}	
\usepackage{bold-extra} 
\usepackage{hyperref}
\usepackage{stmaryrd} 

\newenvironment{mitem}
{\begin{itemize}
  \setlength{\itemsep}{1pt}
  \setlength{\parskip}{0pt}
  \setlength{\parsep}{0pt}}
{\end{itemize}}

\newtheorem{theorem}{Theorem}
\newtheorem*{thm*}{Theorem}
\newtheorem{proposition}[theorem]{Proposition}
\newtheorem*{prop*}{Proposition}
\newtheorem{lemma}[theorem]{Lemma}
\newtheorem{corollary}[theorem]{Corollary}
\theoremstyle{definition}
\newtheorem{definition}{Definition}

\theoremstyle{remark}

\newcommand{\ZX}{\textsc{zx}}
\newcommand{\NOT}{\textsc{not}}
\newcommand{\SWAP}{\textsc{swap}}

\newcommand{\abs}[1]{\left| #1 \right|}
\newcommand{\avg}[1]{\left\langle #1 \right\rangle}

\newcommand{\ZZ}{\mathbb{Z}}

\newcommand{\ket}[1]{\left| #1 \right>} 

\newcommand{\intf}[1]{\left\llbracket #1 \right\rrbracket} 

\usepackage[svgnames]{xcolor} 

\usepackage{tikz}
\usetikzlibrary{shapes.geometric} 
\usetikzlibrary{shapes.misc} 
\usetikzlibrary{shapes.symbols} 
\usetikzlibrary{decorations.pathreplacing} 

\pgfdeclarelayer{edgelayer}
\pgfdeclarelayer{nodelayer}
\pgfsetlayers{edgelayer,nodelayer,main}

\tikzstyle{none}=[inner sep=0pt]
\tikzstyle{rn}=[circle,fill=Red,draw=Black,line width=0.8 pt,minimum size=5pt,inner sep=0pt]
\tikzstyle{gn}=[circle,fill=Lime,draw=Black,line width=0.8 pt,minimum size=5pt,inner sep=0pt]
\tikzstyle{Hadamard}=[rectangle,fill=Yellow,draw=Black,minimum size=8pt,inner sep=0pt,label={center:$\scriptstyle\mathrm{H}$}]
\tikzstyle{HadSpek}=[rectangle,fill=Yellow,draw=Black,minimum size=6pt]
\tikzstyle{gphase}=[rounded rectangle,rounded rectangle arc length=90,fill=Lime!20,inner sep=2pt]
\tikzstyle{rphase}=[rounded rectangle,rounded rectangle arc length=90,fill=Red!20,inner sep=2pt]
\tikzstyle{bphase}=[rounded rectangle,rounded rectangle arc length=90,fill=Black!20,inner sep=2pt]
\tikzstyle{normalrect}=[rectangle,fill=white,draw=black,minimum height=12pt,minimum width=12pt,inner sep=0pt]

\tikzstyle{bn}=[circle,fill=black,draw=black,inner sep=1pt]
\tikzstyle{whn}=[circle,fill=white,draw=black,inner sep=1pt]
\tikzstyle{bigcloud}=[cloud,fill=white,draw=black]
\tikzstyle{diagram}=[rectangle,fill=white,draw=black,minimum height=0.5cm,minimum width=1cm,inner sep=0pt]

\tikzstyle{dagger_morphism}=[trapezium,trapezium left angle=90,trapezium right angle=60,fill=white,draw=black,minimum height=14pt,minimum width=12pt,inner sep=0pt]
\tikzstyle{morphism}=[trapezium,trapezium left angle=90,trapezium right angle=120,fill=white,draw=black,minimum height=14pt,minimum width=12pt,inner sep=0pt]
\tikzstyle{dualmorphism}=[trapezium,trapezium right angle=90,trapezium left angle=60,fill=white,draw=black,minimum height=14pt,minimum width=12pt,inner sep=0pt]
\tikzstyle{dtr}=[regular polygon,regular polygon sides=3,shape border rotate=180,fill=white,draw=black,minimum size=22pt,inner sep=0pt]
\tikzstyle{utr}=[regular polygon,regular polygon sides=3,fill=white,draw=black,minimum size=22pt,inner sep=0pt]
\tikzstyle{bigdtr}=[isosceles triangle,shape border rotate=-90,fill=white,draw=black,isosceles triangle stretches,minimum height=22pt,minimum width=2cm,inner sep=0pt]

\tikzstyle{squ}=[rectangle,fill=white,draw=black,minimum height=0.5cm,minimum width=0.6cm,inner sep=0pt]
\tikzstyle{twqu}=[rectangle,fill=white,draw=black,minimum height=1.25cm,minimum width=0.6cm,inner sep=0pt]
\tikzstyle{thrqu}=[rectangle,fill=white,draw=black,minimum height=2cm,minimum width=0.6cm,inner sep=0pt]
\tikzstyle{bn}=[circle,fill=black,draw=black,inner sep=0pt,minimum size=4pt]
\tikzstyle{whn}=[circle,fill=none,draw=black,minimum size=8pt]
\tikzstyle{xn}=[circle,fill=none,draw=none,minimum size=8pt]

\newcommand{\effect}[1]{\begin{tikzpicture}[baseline=.1cm]
	\begin{pgfonlayer}{nodelayer}
		\node [style=none] (0) at (0, 0) {};
		\node [style={#1}] (1) at (0, 0.25) {};
	\end{pgfonlayer}
	\begin{pgfonlayer}{edgelayer}
		\draw [] (0.center) to (1);
	\end{pgfonlayer}
\end{tikzpicture}}

\newcommand{\scalar}[1]{\begin{tikzpicture}[baseline=-.1cm]
	\begin{pgfonlayer}{nodelayer}
		\node [style={#1}] (0) at (0, 0) {};
	\end{pgfonlayer}
\end{tikzpicture}}

\newcommand{\state}[1]{\begin{tikzpicture}[baseline=-.1cm]
	\begin{pgfonlayer}{nodelayer}
		\node [style={#1}] (0) at (0, 0) {};
		\node [style=none] (1) at (0, 0.25) {};
	\end{pgfonlayer}
	\begin{pgfonlayer}{edgelayer}
		\draw (0) to (1.center);
	\end{pgfonlayer}
\end{tikzpicture}}

\newcommand{\phase}[1]{\begin{tikzpicture}[baseline=-.1cm]
	\begin{pgfonlayer}{nodelayer}
		\node [style=none] (0) at (0, 0.2) {};
		\node [style={#1}] (1) at (0, -0) {};
		\node [style=none] (2) at (0, -0.2) {};
	\end{pgfonlayer}
	\begin{pgfonlayer}{edgelayer}
		\draw [] (0.center) to (1);
		\draw [] (1) to (2.center);
	\end{pgfonlayer}
\end{tikzpicture}}

\newcommand{\splitnode}[1]{\begin{tikzpicture}[baseline=-0.1cm]
	\begin{pgfonlayer}{nodelayer}
		\node [style={#1}] (0) at (0, -0) {};
		\node [style=none] (1) at (0, -0.25) {};
		\node [style=none] (2) at (-0.25, 0.25) {};
		\node [style=none] (3) at (0.25, 0.25) {};
	\end{pgfonlayer}
	\begin{pgfonlayer}{edgelayer}
		\draw [bend right=15] (2.center) to (0);
		\draw [bend left=15] (3.center) to (0);
		\draw (0) to (1.center);
	\end{pgfonlayer}
\end{tikzpicture}}

\newcommand{\joinnode}[1]{\begin{tikzpicture}[baseline=-0.1cm]
	\begin{pgfonlayer}{nodelayer}
		\node [style={#1}] (0) at (0, -0) {};
		\node [style=none] (1) at (0, 0.25) {};
		\node [style=none] (2) at (0.25, -0.25) {};
		\node [style=none] (3) at (-0.25, -0.25) {};
	\end{pgfonlayer}
	\begin{pgfonlayer}{edgelayer}
		\draw [bend right=15] (2.center) to (0);
		\draw [bend left=15] (3.center) to (0);
		\draw (0) to (1.center);
	\end{pgfonlayer}
\end{tikzpicture}}

\newcommand{\HadSpek}[0]{\begin{tikzpicture}[baseline=-0.15cm]
	\begin{pgfonlayer}{nodelayer}
		\node [style=HadSpek] (0) at (0, 0) {};
		\node [style=none] (1) at (0, 0.25) {};
		\node [style=none] (2) at (0, -0.25) {};
	\end{pgfonlayer}
	\begin{pgfonlayer}{edgelayer}
		\draw (0.center) to (1);
		\draw (2.center) to (1);
	\end{pgfonlayer}
\end{tikzpicture}}

\newcommand{\gendiagram}[1]{\begin{tikzpicture}[scale=0.5,baseline=-0.1cm]
	\begin{pgfonlayer}{nodelayer}
		\node [style=diagram] (0) at (0, -0) {#1};
		\node [style=none] (1) at (-0.75, 1) {};
		\node [style=none] (2) at (0.75, 1) {};
		\node [style=none] (3) at (-0.75, -1) {};
		\node [style=none] (4) at (0.75, -1) {};
		\node [style=none] (5) at (0, 0.75) {$\ldots$};
		\node [style=none] (6) at (0, -0.75) {$\ldots$};
	\end{pgfonlayer}
	\begin{pgfonlayer}{edgelayer}
		\draw (1.center) to (3.center);
		\draw (2.center) to (4.center);
	\end{pgfonlayer}
\end{tikzpicture}}

\newlength{\Speklength}
\newcommand{\Spekcolour}[0]{teal}

\newcommand{\SpekOneSquare}[4]{\begin{tikzpicture}[baseline=-0.1cm,fill=\Spekcolour, every path/.style={draw}]
  #1 (0,-\Speklength) rectangle +(\Speklength,\Speklength);
  #2 (\Speklength,-\Speklength) rectangle +(\Speklength,\Speklength);
  #3 (0,0) rectangle +(\Speklength,\Speklength);
  #4 (\Speklength,0) rectangle +(\Speklength,\Speklength);
\end{tikzpicture}}

\newcommand{\SpekTwoSquare}[1]{\begin{tikzpicture}[baseline=-0.1cm,fill=\Spekcolour]
  #1%
  \foreach \x in {-2\Speklength,-\Speklength,0,\Speklength}
    \foreach \y in {-2\Speklength,-\Speklength,0,\Speklength}
      \draw (\x,\y) rectangle +(\Speklength,\Speklength);
\end{tikzpicture}}

\newcommand{\Spekfill}[2]{\fill (#1\Speklength,#2\Speklength) rectangle +(\Speklength,\Speklength);
}

\title{A complete graphical calculus for Spekkens' toy bit theory}
\author{Miriam Backens
\institute{Department of Computer Science,\\ University of Oxford,\\ Oxford, United Kingdom}
\email{miriam.backens@cs.ox.ac.uk}
\and
Ali Nabi Duman
\institute{Department of Mathematics,\\ King Fahd University of Petroleum \& Minerals,\\ Dhahran, Saudi Arabia}
\email{alinabi@gmail.com}
}
\def\titlerunning{A complete graphical calculus for Spekkens' toy theory}
\def\authorrunning{Miriam Backens \& Ali Nabi Duman}

\begin{document}

\maketitle

\begin{abstract}
 While quantum theory cannot be described by a local hidden variable model, it is nevertheless possible to construct such models that exhibit features commonly associated with quantum mechanics.
 These models are also used to explore the question of $\psi$-ontic versus $\psi$-epistemic theories for quantum mechanics.
 Spekkens' toy theory is one such model.
 It arises from classical probabilistic mechanics via a limit on the knowledge an observer may have about the state of a system.
 The toy theory for the simplest possible underlying system closely resembles stabilizer quantum mechanics, a fragment of quantum theory which is efficiently classically simulable but also non-local.
 Further analysis of the similarities and differences between those two theories can thus yield new insights into what distinguishes quantum theory from classical theories, and $\psi$-ontic from $\psi$-epistemic theories.

 In this paper, we develop a graphical language for Spekkens' toy theory.
 Graphical languages offer intuitive and rigorous formalisms for the analysis of quantum mechanics and similar theories.
 To compare quantum mechanics and a toy model, it is useful to have similar formalisms for both.
 We show that our language fully describes Spekkens' toy theory and in particular, that it is complete: meaning any equality that can be derived using other formalisms can also be derived entirely graphically.
 Our language is inspired by a similar graphical language for quantum mechanics called the ZX-calculus.
 Thus Spekkens' toy bit theory and stabilizer quantum mechanics can be analysed and compared using analogous graphical formalisms.
\end{abstract}

\section{Introduction}

The study of the differences between quantum physical behaviour and classical behaviour is at the heart of much foundational research in quantum physics.
The usual way of analysing these differences is by finding phenomena that are intrinsic to one theory and do not appear in the other, for example the violation of Bell inequalities \cite{bell_einstein-podolsky-rosen_1964}.
Yet there is also another approach: that of building toy models which reproduce phenomena generally considered quantum even though their description is entirely rooted in classical physics.

One such toy model is Spekkens' toy theory, first introduced in \cite{spekkens_evidence_2007}.
The original description of the model was informal, but it was put into equational form in \cite{coecke_phase_2011}.
Furthermore, the toy model was redefined in a rigorous way in \cite{spekkens_quasi-quantization_2014}; this redefinition is consistent with \cite{coecke_phase_2011}.
Despite being a local hidden variable theory, Spekkens' toy theory reproduces many features of quantum mechanics, e.g. incompatibility of certain observables, teleportation, and no-cloning.
The toy theory for the simplest kind of system -- the toy bit theory -- closely resembles the theory of stabilizer quantum mechanics \cite{spekkens_evidence_2007,pusey_stabilizer_2012}.
Stabilizer quantum mechanics is a fragment of quantum theory resulting from a restriction of the allowed operations to preparation of states in the computational basis, unitary Clifford operations, and computational basis measurements \cite{gottesman_stabilizer_1997}.

There are also some phenomena that appear in stabilizer quantum theory but are not replicated in the toy model, e.g.\ the above-mentioned violation of Bell inequalities.
This is related to the fact that Spekkens' toy theory is a $\psi$-epistemic theory by construction.
A $\psi$-epistemic theory is a theory where the state that an observer assigns to a system, is not real: it is only an artefact of the restricted knowledge of the observer.
Quantum theory on the other hand is considered to be $\psi$-ontic, i.e.\ it is a theory where the states an observer assigns to a system are real \cite{pusey_reality_2012}.

To compare and contrast the two theories -- Spekkens' toy bit theory and stabilizer quantum theory -- in more detail, it is useful to have similar mathematical formalisms for describing both.
Particularly, what are needed are \emph{high-level} formalisms, which hide some of the details of the underlying theories so as to focus on conceptual properties.
Graphical languages are high-level languages that use two-dimensional diagrams.
These two-dimensional languages allow \emph{parallel composition} -- applying transformations to two different systems at the same time -- to be separated from \emph{sequential composition} -- the application of transformations to the same system at different times -- by designating one dimension to roughly correspond to ``space'' and the other to ``time''.
A major difference between classical physics and quantum physics is the way the state spaces of systems compose in parallel, i.e.\ when the systems are put ``side by side'' \cite{abramsky_categorical_2008}: classically, the resulting state space is the Cartesian product of the original spaces, meaning each state of the joint system can be described by specifying separate states for each of the component systems.
For quantum systems, on the other hand, the joint state space is the tensor product of the original state spaces and joint states may not correspond to well-defined states of the separate systems: they can be entangled.
Thus in the study of quantum foundations, the study of composite systems is central, and graphical languages offer an intuitive way of doing that.

Spekkens' toy bit theory and stabilizer quantum theory have previously been studied using the stabilizer formalism \cite{pusey_stabilizer_2012}.
The two theories have also been compared using methods from categorical quantum mechanics \cite{coecke_toy_2011}.
Categorical quantum mechanics is the analysis of quantum mechanics via the mathematical formalisms of category theory, pioneered by Abramsky and Coecke \cite{abramsky_categorical_2004}.
While categorical quantum mechanics has given rise to a range of graphical languages for quantum theory, Spekkens' toy theory has not been analysed graphically before.
The \ZX-calculus is one such graphical language based on categorical quantum mechanics \cite{coecke_interacting_2008,coecke_interacting_2011}; we use it as inspiration for the construction of a graphical calculus for the toy theory.

For a graphical formalism to capture a model fully, it needs to satisfy several properties.
Firstly, it should be \emph{universal}, i.e.\ it should be possible to represent graphically any process allowed in the model.
The formalism should furthermore be \emph{sound}, i.e.\ any equality that can be derived graphically should be derivable using other standard formalisms for the model.
Lastly, the graphical formalism should be \emph{complete}, meaning that any equality that can be derived using other standard formalisms can also be derived graphically.
The \ZX-calculus is universal, sound \cite{coecke_interacting_2011}, and complete \cite{backens_zx-calculus_2013} for pure state qubit stabilizer quantum mechanics with post-selected measurements.
We show that our graphical calculus satisfies all three of these properties for the maximal knowledge fragment of the toy theory (which corresponds to pure states in quantum theory) with post-selected measurements.

We introduce Spekkens' toy theory and stabilizer quantum mechanics in section \ref{s:theories}.
Section \ref{s:background} contains an overview over graphical languages and how to make them rigorous.
We then define the graphical calculus for Spekkens' toy bit theory in section \ref{s:toy_theory} and prove that it is universal and sound.
The completeness proof for the graphical calculus is given in section \ref{s:graphical_calculus}, with conclusions in section \ref{s:conclusions}.

\section{Spekkens' toy theory and stabilizer quantum mechanics}
\label{s:theories}

In this section we introduce Spekkens' toy bit theory as well as stabilizer quantum mechanics.
We also present some of the standard formalisms used to analyse the two theories.

\subsection{Definition of Spekkens' toy bit theory}
\label{s:spekkens}

Spekkens' toy theory is a local hidden variable theory that nevertheless displays many of the same properties and effects as quantum mechanics \cite{spekkens_evidence_2007}.
It arises from classical probabilistic mechanics via an epistemic restriction, i.e.\ a restriction on the knowledge an observer may have about the state of the system \cite{spekkens_quasi-quantization_2014}.

We shall only consider the toy theory for the simplest non-trivial system here: the toy bit theory.
A single toy bit is a system with four states, these are the \emph{ontic states} or states of reality.
These states are often drawn on a 2 by 2 grid as in Fig. \ref{fig:toy_bit_state_space} a.
An ontic state can be described by giving the values -- 0 or 1 -- for two variables $X$ and $Z$.

\begin{figure}
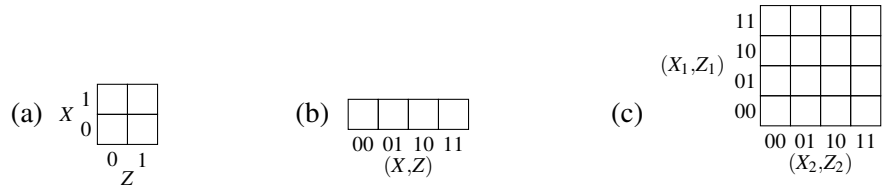

 \centering
 (a) \input{tikz_files/spekkens_state_square.tikz} $\qquad\qquad$
 (b) $\;$\input{tikz_files/spekkens_state_line.tikz} $\qquad\qquad$
 (c) \input{tikz_files/spekkens_2state.tikz}
 \caption{(a) \& (b) Visualisations of the state space of a single toy bit. (c) Visualisation of the joint state space of two toy bits. Specific states can be represented by colouring in cells in the diagram.}
 \label{fig:toy_bit_state_space}
\end{figure}

An observer or experimenter working with toy bits does not have direct access to the ontic states, instead they assign to a system an \emph{epistemic state}, a state of knowledge.
The observer can learn about the state of a system by measuring \emph{quadrature variables}, which are linear combinations of the variables $X$ and $Z$ -- for a single toy bit, these are $X$, $Z$, or $X\oplus Z$, where $\oplus$ denotes addition modulo 2.
As in quantum mechanics, the quadrature variables $X$ and $Z$ for the same toy bit are considered non-commuting: this is done by imposing a commutation relation $[\cdot,\cdot]$ satisfying:
\begin{equation}
 [X,X]=0=[Z,Z] \qquad \text{and} \qquad [X,Z] = 1 = [Z,X],
\end{equation}
which is furthermore linear, so that e.g.:
\begin{equation}
 [X\oplus Z,Z] = 1.
\end{equation}
Now the knowledge an observer may have is determined by the \emph{principle of classical complementarity} \cite{spekkens_quasi-quantization_2014}:
\begin{quotation}
 \noindent The valid epistemic states are those where an agent knows the values of a set of commuting quadrature variables and is ignorant otherwise.
\end{quotation}
\emph{States of maximal knowledge} are those epistemic states where the observer knows the values of a maximal set of commuting quadrature variables, i.e.\ a set with which no other quadrature variable commutes.
These states correspond to pure states in quantum theory, and we will only consider states of maximal knowledge in this paper.

There are six states of maximal knowledge of a single toy bit.
Single-toy bit states can be visualised on the 2 by 2 grid by colouring in those ontic states that are consistent with the knowledge of the observer.
For example, the state $X\oplus Z=0$ is shown in Fig. \ref{fig:toy_bit_state_examples} a.

\begin{figure}
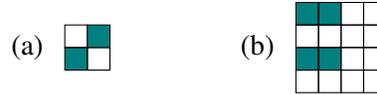

 \centering
 (a) $\;$ \SpekOneSquare{\fill}{\draw}{\draw}{\fill} $\qquad\qquad$
 (b) $\;$ \SpekTwoSquare{\Spekfill{-2}{-1} \Spekfill{-1}{-1} \Spekfill{-2}{1} \Spekfill{-1}{1}}
 \caption{(a) The single-toy bit state characterised by $X\oplus Z=0$. (b) The joint state of two toy bits corresponding to $Z_1=1\wedge X_2=0$.}
 \label{fig:toy_bit_state_examples}
\end{figure}

Multiple toy bits can be considered jointly, in which case the variables $X$ and $Z$ for separate subsystems are considered to commute, i.e.:
\begin{equation}
 [X_i,Z_j] = \delta_{ij},
\end{equation}
where the subscripts denote the subsystem to which the variable belongs and $\delta_{ij}$ is 1 if $i=j$ and 0 otherwise.
Thus, e.g.\ $Z_1=1\wedge X_2=0$ is an example of a valid joint state of two toy bits.
Joint states of two toy bits can be visualised by stretching the diagram for one bit into a line (see Fig. \ref{fig:toy_bit_state_space} b) and then combining two such lines into a 4 by 4 grid as shown in Fig. \ref{fig:toy_bit_state_space} c. 
The above-mentioned state is drawn in Fig. \ref{fig:toy_bit_state_examples} b.

The valid reversible transformations in the toy theory are those permutations of the ontic states that map all valid epistemic states to valid epistemic states.
Any set of commuting quadrature variables is a valid measurement.

As a different notation, the ontic states of the toy theory are sometimes numbered 1 through 4 in the order in which they appear in Fig. \ref{fig:toy_bit_state_space} b \cite{coecke_spekkens_2011}.
Epistemic states can then be denoted by sets, e.g.\ the state $X\oplus Z=0$ corresponds to the set $\{1,4\}$ and the two-toy bit state $Z_1=1\wedge X_2=0$ corresponds to:
\begin{equation}
 \{(2,1),(2,2),(4,1),(4,2)\}.
\end{equation}
Rather than considering states of $n$ toy bits to be sets, they can also be seen as relations from the one-element set $I=\{\bullet\}$ into $IV^n$, the $n$-fold Cartesian product of the set $IV=\{1,2,3,4\}$.
Formally, a relation $R$ between sets $A$ and $B$ is a subset of the Cartesian product $A\times B$.
We use the notation $a\sim b$ to indicate that $(a,b)\in R$.
Thus, the state $X\oplus Z=0$ shown in Fig. \ref{fig:toy_bit_state_examples} a corresponds to the relation:
\begin{equation}
 \bullet \sim \{1,4\}.
\end{equation}
Similarly, post-selected measurements on $n$ toy bits can be seen as relations from $IV^n$ to $I$, e.g.\ the single-toy bit measurement of the $Z$ variable with outcome 1 corresponds to:
\begin{equation}
 \{2,4\}\sim\bullet.
\end{equation}
Reversible transformations can also be considered as relations.
This perspective puts state preparation and post-selected measurements on an equal footing with reversible transformations and allows any process on toy bits to be considered as a relation.

In a slight abuse of notation, we use states and states-as-relations interchangeably.

\subsection{Stabilizer quantum mechanics}
\label{s:stabilizer_QM}

Stabilizer quantum mechanics is a restriction of the full quantum theory.
It includes only those states that are simultaneous eigenstates of several tensor products of Pauli operators, and unitary transformations that map this set of states back to itself.
This theory was first introduced in the context of error-correcting codes \cite{gottesman_stabilizer_1997}.

In general, $2^n$ complex numbers are required to specify a quantum state on $n$ qubits -- these can be, for example, the components of the vector describing the state in terms of the computational basis.
For stabilizer states, there exists a more efficient description by specifying a generating set for the group of Pauli products that \emph{stabilizes} the state, i.e.\ maps it back to itself.
As an example, consider the Bell state:
\begin{equation}
 \frac{1}{\sqrt{2}}(\ket{00}+\ket{11}),
\end{equation}
which is a stabilizer state.
It is stabilized by the following group of Pauli products:
\begin{equation}
 \{ I\otimes I, X\otimes X, Z\otimes Z, -Y\otimes Y\},
\end{equation}
which is generated e.g.\ by:
\begin{equation}\label{eq:Bell_stabilizer}
 \avg{X\otimes X, Z\otimes Z}.
\end{equation}

There exists a representation of pure $n$-qubit stabilizer states in terms of $2n$ by $n$ binary matrices called \emph{check matrix} \cite{calderbank_quantum_1997}.
Each column in the matrix corresponds to one of the Pauli products generating the stabilizer group, ignoring the factor of $\pm 1$.
A single Pauli matrix is encoded in two bits as follows:
\begin{align}
 I &\mapsto 00, \\
 X &\mapsto 01, \\
 Y &\mapsto 11, \text{ and} \\
 Z &\mapsto 10.
\end{align}
For a Pauli product on $n$ qubits, the $m$-th factor in the tensor product is represented by the $m$-th and $(m+n)$-th components of the vector.
Thus the generating set from \eqref{eq:Bell_stabilizer} yields the following check matrix:
\begin{equation}
 \begin{pmatrix} 0&1 \\ 0&1 \\ 1&0 \\ 1&0 \end{pmatrix},
\end{equation}
where the first column represents $X\otimes X$ and the second column $Z\otimes Z$.

The pure state qubit stabilizer fragment can also be defined operationally in terms of the following transformations \cite{nielsen_quantum_2010}:
\begin{itemize}
 \item preparation of states in the computational basis
 \item unitary Clifford operations, generated by the single-qubit Hadamard operator, $H$, and the phase operator, $S$:
  \[
   H = \frac{1}{\sqrt{2}} \begin{pmatrix}1&1\\1&-1\end{pmatrix} \quad \text{and} \quad S = \begin{pmatrix}1&0\\0&i\end{pmatrix},
  \]
  as well as the two-qubit controlled-\NOT{} operator:
  \begin{equation}\label{eq:controlled-NOT}
   C_X = \begin{pmatrix} 1&0&0&0 \\ 0&1&0&0 \\ 0&0&0&1 \\0&0&1&0 \end{pmatrix}
  \end{equation}
 \item measurements in the computational basis.
\end{itemize}
We make use of both the operational description and the binary formalism in later parts of this paper.

\section{Graphical languages}
\label{s:background}

Graphical languages provide an intuitive and high-level way of reasoning.
One of the reasons for this is the fact that they allow parallel and sequential composition to be denoted along two different dimensions.
By parallel and sequential composition, we mean that there are two different ways of composing transformations:
Parallel composition corresponds, roughly, to applying transformations to different (e.g.\ spatially separated) systems at the same time.
Sequential composition, on the other hand, corresponds to applying transformations to the same system at different times.

In linear algebraic notation, parallel and sequential composition are usually distinguished by different operator symbols.
For example, the parallel composition of two matrices $A$ and $B$ is the tensor product, denoted by $A\otimes B$.
The sequential composition of two matrices $A$ and $B$ is the matrix product, usually written simply as $AB$.
Long algebraic expressions can thus become difficult to parse: it is not easy to see how different components of the expression compose.

In a graphical language, parallel composition can be denoted by stacking symbols vertically, and sequential composition by juxtaposing them horizontally (or conversely, depending on convention).
An example of such a graphical notation is the quantum circuit notation \cite{deutsch_quantum_1989,nielsen_quantum_2010}.
With a two-dimensional notation, it is much easier to see how components compose, even in long and complicated expressions.
For example, the quantum circuit in Fig.\ \ref{fig:circuit_example} can be written algebraically as:
\begin{equation}\label{eq:circuit_example}
 (I\otimes I\otimes (U_4 U_3)) W (I\otimes V) (U_1\otimes U_2\otimes I),
\end{equation}
where $I$ denotes the single-qubit identity operator.
The diagram is much easier to take in at a glance than the algebraic expression.

\begin{figure}
 \centering
 \input{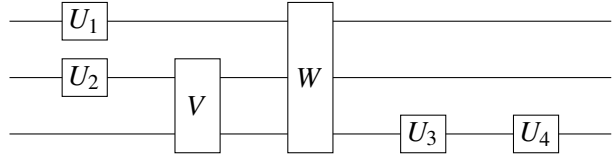}
 \caption{A quantum circuit diagram on three qubits. The operators $U_1,U_2,U_3,$ and $U_4$ are single-qubit unitaries, $V$ is a two-qubit unitary, and $W$ a three-qubit unitary.}
 \label{fig:circuit_example}
\end{figure}

Yet graphical languages are often introduced as informal personal short-hands and used to develop an intuitive understanding of a problem that can then be confirmed using a more rigorous but less intuitive language.
This means doing the same work twice: once graphically, then again in the alternative formalism.
To avoid this, the graphical languages need to be made rigorous; then reliable results can be derived entirely graphically.

\subsection{Making graphical languages rigorous}
\label{s:graphical_rigorous}

There are two steps to the process of making graphical languages rigorous: firstly, one needs to give an explicit translation between diagrams and algebraic terms. Secondly, one needs to prove that two diagrams that seem intuitively equal translate to algebraic terms that are equal.

We are here considering languages that are similar to quantum circuits in that they consist of boxes, which denote transformations, and wires, which correspond to systems or the identity transformations on those systems.

The natural formalism for making graphical languages rigorous is category theory, as \emph{monoidal categories} are the most general mathematical structures incorporating both parallel and sequential composition of transformations.
This research programme was begun by Joyal and Street, who used the theory of monoidal categories to give rigorous underpinnings to a range of graphical notations from Feynman diagrams to Petri Nets \cite{joyal_geometry_1991}.
An introduction to category theory aimed at physicists, can be found in \cite{coecke_categories_2010}; the standard textbook is \cite{mac_lane_categories_1998}.

There are other types of categories with additional structure which can be used to model more complicated graphical languages.
For example, \emph{symmetric monoidal categories} include a map that swaps two systems that have been composed in parallel.
These categories can be used to model graphical languages with wire crossings where ``wires don't tangle'': i.e.\ two wire crossings are the same as no crossing at all.
An example of this are quantum circuits, where wire crossings are used to denote the transformation interchanging two qubits, for example the circuit shown in Fig. \ref{fig:swap}.

\begin{figure}
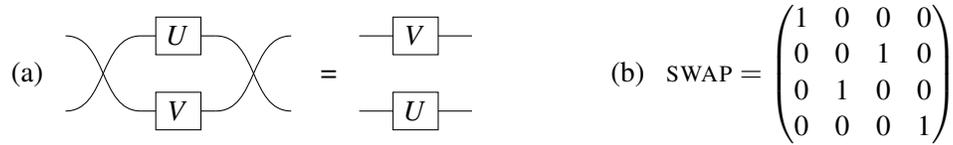

 \centering
 (a) $\;$ \input{tikz_files/circuit_swap1.tikz} $\;$ = $\;$ \input{tikz_files/circuit_swap2.tikz} $\qquad\qquad$
 (b) $\;$ \SWAP{} $= \begin{pmatrix} 1&0&0&0 \\ 0&0&1&0 \\ 0&1&0&0 \\ 0&0&0&1 \end{pmatrix}$
 \caption{(a) A quantum circuit equality where a wire crossing denotes the \SWAP{} map and $U,V$ are arbitrary single-qubit unitaries. (b) The \SWAP{} matrix.}
 \label{fig:swap}
\end{figure}

Quantum mechanics and similar theories are modelled as \emph{dagger compact closed categories} \cite{abramsky_categorical_2004}.
These are categories which have a dagger map -- a generalisation of the Hermitian adjoint for linear maps -- and a compact structure, which corresponds to completely entangled states and measurement effects.
Graphically, the compact structure is denoted by curved wires -- ``cups'' and ``caps'' in a language which is read from bottom to top -- which satisfy the \emph{snake equations} shown in Fig. \ref{fig:compact_structure}.
In quantum theory, the ``cup'' can be thought of as the preparation of a completely entangled state on two systems; the ``cap'' is the outcome of finding that same state when doing a joint measurement of two systems.
The snake equations thus correspond to a post-selected formulation of quantum teleportation.

\begin{figure}
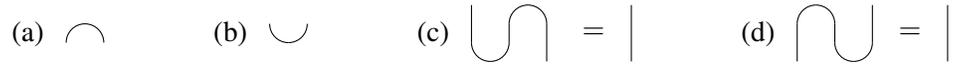

 \centering
 (a) $\;$ \input{tikz_files/cap_diagram.tikz} $\qquad\quad$
 (b) $\;$ \input{tikz_files/cup_diagram.tikz} $\qquad\quad$
 (c) $\;$ \input{tikz_files/snake1.tikz} $\qquad\quad$
 (d) $\;$ \input{tikz_files/snake2.tikz}
 \caption{(a) \& (b) The cap and cup denoting a compact structure. (c) \& (d) The snake equations satisfied by caps and cups. Note that this graphical notation is read from bottom to top rather than left-to-right.}
 \label{fig:compact_structure}
\end{figure}

In graphical languages based on dagger compact closed categories, two diagrams represent the same map if they are equal up to topological transformations that keep the inputs and outputs of the diagrams as a whole invariant \cite{selinger_dagger_2007}.
Topological transformations here are operations like lengthening or shortening wires, bending or straightening wires, or moving boxes around while keeping their connections the same.
Both the equality in Fig. \ref{fig:swap} and the snake equations in Fig. \ref{fig:compact_structure} are examples of such topological transformations.

Yet topological transformations are not enough to yield all desired equalities between diagrams: e.g.\ the diagram equality in Fig. \ref{fig:circuit_rule} is true, but it is not a topological transformation, and cannot be made into one by a change of notation either.
By assuming a set of axiomatic equalities called \emph{rewrite rules}, further graphical equalities can be derived using graphical rewriting.
The idea is that whenever two diagrams are equal, any time one of those diagrams appears as part of a larger diagram, it can be ``cut out'' and replaced by the other diagram.
We gloss over the details of this rewrite process here and assume that a simple copy-paste process works.
An example using quantum circuits is given in Fig. \ref{fig:circuit_rule} and Fig. \ref{fig:circuit_rewrite}.

\begin{figure}
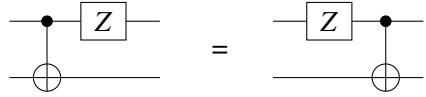

 \centering
 \input{tikz_files/cNOT-Z.tikz} $\quad$ = $\quad$ \input{tikz_files/Z-cNOT.tikz}
 \caption{A quantum circuit equality, where $Z$ denotes the Pauli-Z gate and the other symbol is the controlled-\NOT{} gate as defined in \eqref{eq:controlled-NOT}.}
 \label{fig:circuit_rule}
\end{figure}

\begin{figure}
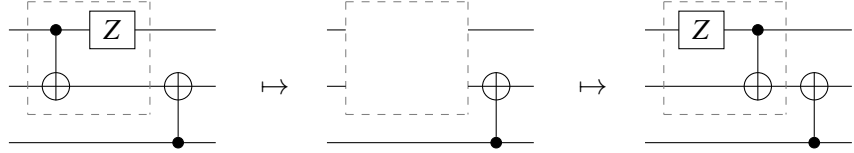

 \centering
 \input{tikz_files/circuit_cNOT-Z-cNOT.tikz} $\quad \mapsto \quad$ \input{tikz_files/circuit_cNOT-Z-cNOT_cut.tikz} $\quad \mapsto \quad$ \input{tikz_files/circuit_Z-cNOT-cNOT.tikz}
 \caption{Rewriting a quantum circuit using the rewrite rule given in Fig. \ref{fig:circuit_rule}: the left-hand side of that equality matches part of the diagram here, that part is ``cut out'', and the right-hand side of Fig. \ref{fig:circuit_rule} is pasted in.}
 \label{fig:circuit_rewrite}
\end{figure}

This process of introducing rewrite rules is fairly arbitrary, posing the question of which equalities should become rewrite rules and how many rewrite rules are necessary.
That issue can be simplified by considering the graphical calculus as a formal system.

\subsection{Graphical calculi as formal systems}
\label{s:formal_systems}

A graphical calculus with rewrite rules can be considered as a formal system that allows the derivation of equalities from the axioms given by the rewrite rules.
For such a calculus to be a useful alternative to a more standard algebraic language, it needs to satisfy certain properties, which are all relative to the interpretation of the diagrams.
Given a diagram $D$, we denote by $\intf{D}$ its interpretation, i.e.\ the process corresponding to $D$.
The desired properties of a graphical language include universality, soundness, and completeness, which are defined as follows:
\begin{mitem}
 \item A graphical calculus for a theory is \emph{universal} if any state or transformation allowed by the theory can be represented graphically, i.e.\ for any process $P$ there exists a diagram $D$ such that:
  \begin{equation}
   \intf{D} = P.
  \end{equation}
 \item It is \emph{sound} if any equality that can be derived graphically can also be derived in the underlying theory, i.e.\ if for any two diagrams $D_1$ and $D_2$:
  \begin{equation}
   D_1 = D_2 \implies \intf{D_1} = \intf{D_2}.
  \end{equation}
 \item A graphical calculus is \emph{complete} if any equality that can be derived in the underlying theory can also be derived graphically, i.e.\ if for any two diagrams $D_1$ and $D_2$:
  \begin{equation}
   \intf{D_1} = \intf{D_2} \implies D_1 = D_2.
  \end{equation}
\end{mitem}

Universality and soundness are usually straightforward to determine:
The property of universality does not involve the rewrite rules at all, it only depends on the diagram components and their interpretations.
Soundness does depend on the rewrite rules, but it can be checked on a rule-by-rule basis: if each rule is sound, then the graphical calculus as a whole is sound.
Completeness is more difficult to prove because it relies on the interaction of all the rewrite rules.

\section{A graphical calculus for Spekkens' toy bit theory}
\label{s:toy_theory}

In this section, we construct a graphical calculus for the maximal knowledge fragment of Spekkens' toy bit theory with post-selected measurements.
This graphical calculus is modelled after the \ZX-calculus for quantum mechanics \cite{coecke_interacting_2008}.
We define basic elements of the graphical notation and show how to combine them into more complicated diagrams.
Next, we give rewrite rules for those diagrams.
We argue that the calculus is universal and sound for Spekkens' toy bit theory.
Finally, we compare the graphical calculus for the toy theory to the \ZX-calculus.

\subsection{Components of the toy theory graphical calculus}

The graphical calculus for the toy theory is read from bottom to top.
Rather than denoting maps by labelled boxes, most maps are denoted by circular nodes which may have labels attached.
As before, we use $\intf{D}$ to denote the process corresponding to a diagram $D$.
 
Define \splitnode{gn} to be the following map from one toy bit to two toy bits:
\begin{equation}
 \intf{\splitnode{gn}} := \begin{cases}
    1 \sim \{(1,1),(2,2)\} \\
    2 \sim \{(1,2),(2,1)\} \\
    3 \sim \{(3,3),(4,4)\} \\
    4 \sim \{(3,4),(4,3)\} .
   \end{cases}
\end{equation}
This is a valid process in the toy theory; it can be considered to consist of the preparation of an ancilla in some fixed state followed by some joint reversible operation on the original toy bit and the ancilla.

The \emph{converse} of a relation $R$, denoted by $R^\dagger$, is defined as:
\begin{equation}
 R^\dagger = \{(b,a)\mid (a,b)\in R\}.
\end{equation}
Let \joinnode{gn} be the converse of \splitnode{gn}:
\begin{equation}
 \intf{\joinnode{gn}} := \intf{\splitnode{gn}}^\dagger.
\end{equation}
This is also a valid process in the toy theory, which can be thought of as a reversible operation on two toy bits, followed by a post-selected measurement of one of them.
As indicated by this notation, the relational converse is the toy theory equivalent of the Hermitian adjoint \cite{coecke_spekkens_2011}.

More complicated diagrams in the toy theory graphical calculus can be built by putting smaller diagrams side-by-side, which corresponds to taking the Cartesian product of the corresponding relations; i.e.\ if:
\begin{center}
 \gendiagram{$D$} $\quad$ and $\quad$ \gendiagram{$D'$}
\end{center}
denote two arbitrary diagrams, then:
\begin{equation}
 \left\llbracket \gendiagram{$D$} \; \gendiagram{$D'$} \right\rrbracket = \left\llbracket \gendiagram{$D$} \right\rrbracket \times \left\llbracket \gendiagram{$D'$} \right\rrbracket.
\end{equation}
Connecting the inputs of some diagram to the outputs of another corresponds to the operation of relational composition: if $R:A\to B$ and $S:B\to C$ are two relations, their composite $S\circ R$ is:
\begin{equation}
 S\circ R = \{(a,c)\mid \exists b\in B \text{ s.t. } (a,b)\in R \wedge (b,c)\in S\}.
\end{equation}
Graphically, assuming the number of outputs of $D$ is equal to the number of inputs of $D'$:
\begin{equation}
 \left\llbracket \input{tikz_files/composite_diagram.tikz} \right\rrbracket = \left\llbracket \gendiagram{$D'$} \right\rrbracket \circ \left\llbracket \gendiagram{$D$} \right\rrbracket.
\end{equation}

We introduce a short-hand notation for specific diagrams built from \splitnode{gn} and \joinnode{gn}: a green node with $n$ inputs and $m$ outputs for positive integers $n,m$ is defined as follows:
\begin{equation}\label{eq:spider}
 \begin{tikzpicture}[decoration=brace,baseline=-0.1cm]
	\begin{pgfonlayer}{nodelayer}
		\node [style=none] (0) at (-1.8, 0.5) {};
		\node [style=none] (1) at (-1.5, 0.45) {$\ldots$};
		\node [style=none] (2) at (-1.2, 0.5) {};
		\node [style=gn] (3) at (-1.5, -0) {};
		\node [style=none] (5) at (-1.8, -0.5) {};
		\node [style=none] (6) at (-1.5, -0.45) {$\ldots$};
		\node [style=none] (7) at (-1.2, -0.5) {};
		\node [style=none] (8) at (-1.9, 0.55) {};
		\node [style=none] (9) at (-1.5, 0.75) {$m$};
		\node [style=none] (10) at (-1.1, 0.55) {};
		\node [style=none] (11) at (-1.9, -0.55) {};
		\node [style=none] (12) at (-1.5, -0.75) {$n$};
		\node [style=none] (13) at (-1.1, -0.55) {};
	\end{pgfonlayer}
	\begin{pgfonlayer}{edgelayer}
		\draw [bend left=15] (3) to (7.center);
		\draw [bend right=15] (3) to (5.center);
		\draw [bend right=15] (0.center) to (3);
		\draw [bend right=15] (3) to (2.center);
		\draw [decorate] (8) to (10);
		\draw [decorate] (13) to (11);
	\end{pgfonlayer}
 \end{tikzpicture} \; :=
 \input{tikz_files/spider_no_phase.tikz}
\end{equation}
Such a node is called a \emph{spider}.

Represent the following four single-toy bit states by green nodes with \emph{phase labels}:
\begin{align}
  \intf{\state{gn,label={[gphase]right:$00$}}} &:= \{1,3\}, \label{eq:state00} \\
  \intf{\state{gn,label={[gphase]right:$01$}}} &:= \{1,4\}, \\
  \intf{\state{gn,label={[gphase]right:$10$}}} &:= \{2,3\}, \quad\text{and} \\
  \intf{\state{gn,label={[gphase]right:$11$}}} &:= \{2,4\}, \label{eq:state11}
\end{align}
and let $\state{gn}$ be short-hand for \state{gn,label={[gphase]right:$00$}}.
These two alternative notations make later definitions consistent.
Spiders can now be given phase labels via the following definition:
\begin{equation}\label{eq:phased_spider}
 \begin{tikzpicture}[decoration=brace,baseline=-0.1cm]
	\begin{pgfonlayer}{nodelayer}
		\node [style=none] (0) at (-1.8, 0.5) {};
		\node [style=none] (1) at (-1.5, 0.45) {$\ldots$};
		\node [style=none] (2) at (-1.2, 0.5) {};
		\node [style=gn,label={[gphase]right:$xy$}] (3) at (-1.5, -0) {};
		\node [style=none] (5) at (-1.8, -0.5) {};
		\node [style=none] (6) at (-1.5, -0.45) {$\ldots$};
		\node [style=none] (7) at (-1.2, -0.5) {};
		\node [style=none] (8) at (-1.9, 0.55) {};
		\node [style=none] (9) at (-1.5, 0.75) {$m$};
		\node [style=none] (10) at (-1.1, 0.55) {};
		\node [style=none] (11) at (-1.9, -0.55) {};
		\node [style=none] (12) at (-1.5, -0.75) {$n$};
		\node [style=none] (13) at (-1.1, -0.55) {};
	\end{pgfonlayer}
	\begin{pgfonlayer}{edgelayer}
		\draw [bend left=15] (3) to (7.center);
		\draw [bend right=15] (3) to (5.center);
		\draw [bend right=15] (0.center) to (3);
		\draw [bend right=15] (3) to (2.center);
		\draw [decorate] (8) to (10);
		\draw [decorate] (13) to (11);
	\end{pgfonlayer}
 \end{tikzpicture} \; := \; \input{tikz_files/spider_def_Spek.tikz}
\end{equation}
where $x,y\in\{0,1\}$.
Furthermore, spiders without inputs can be defined by composing \state{gn} and a spider with one input.
Let \effect{gn} be the converse of \state{gn}, seen as a relation:
\begin{equation}
 \intf{\effect{gn}} := \begin{cases} 1\sim\bullet \\ 3\sim\bullet. \end{cases}
\end{equation}
Then arbitrary spiders with no outputs can be defined as composites of a one-output spider and \effect{gn}.
In this way, definitions \eqref{eq:spider} and \eqref{eq:phased_spider} can be extended to arbitrary non-negative numbers of inputs and outputs $n$ and $m$.

Let \HadSpek{} be the following reversible single-toy bit operation:
\begin{equation}
  \intf{\HadSpek} := \begin{cases}
    1 \sim 1 \\
    2 \sim 3 \\
    3 \sim 2 \\
    4 \sim 4.
   \end{cases}
 \end{equation}
As a final short-hand, define red spiders as green spiders with copies of \HadSpek{} on all inputs and outputs:
\begin{equation}
 \input{tikz_files/colour_Spek_def.tikz}
\end{equation}

A single straight wire corresponds to the identity relation and a wire crossing is the obvious \SWAP{} relation interchanging the states of the two subsystems.
A ``cup'' is interpreted as follows:
\begin{equation}
 \intf{\input{tikz_files/cup_diagram.tikz}} = \{(1,1),(2,2),(3,3),(4,4)\},
\end{equation}
and the cap is its converse.

As \splitnode{gn}, \joinnode{gn}, \state{gn, label={[gphase]right:$xy$}}, and \effect{gn} are all special cases of green spiders, the graphical calculus can be considered to consist of green phased spiders with $n$ inputs and $m$ outputs, where $n$ and $m$ are now non-negative integers; red phased spiders with arbitrary numbers of inputs and outputs; and \HadSpek{}.

Spiders with exactly one input and one output are called \emph{phase shifts}, these are the only spiders representing reversible operations.
A diagram with neither inputs nor outputs is called a \emph{scalar diagram}; the same terminology is also used for parts of larger diagrams that are disconnected from all inputs or outputs of the large diagram.

\subsection{Rewrite rules of the toy theory graphical calculus}

We postulate the following rewrite rules for the toy theory graphical calculus.
Any rule given here can also be used with the colours red and green swapped.
Rules can furthermore be used upside-down.

\textbf{Spider rule and loop rule}: Two nodes of the same colour can merge if they are connected by an edge, in that case their phase labels combine by bit-wise addition modulo 2.
Self-loops can be removed.
\begin{center}
 \input{tikz_files/spider_Spek.tikz} $\qquad\qquad$ \input{tikz_files/loop_Spek.tikz}
\end{center}
\vskip 0.2cm

\textbf{Identity rule, bialgebra rule, and copy rule}: A node with one input and one output and no phase label (or, equivalently, phase $00$) is the same as an edge.
The bialgebra rule allows a certain pattern of two red and two green nodes to be replaced by just one red and green node.
A node of one colour with one input and two outputs copies the zero phase state of the other colour.
\begin{center}
 \input{tikz_files/identity_Spek.tikz} $\qquad\qquad\quad$  \input{tikz_files/bialgebra2.0.tikz} $\qquad\qquad\quad$ \input{tikz_files/copy2.0.tikz}
\end{center}
\vskip 0.2cm

\textbf{$11$-copy rule and $11$-commutation rule}: A $11$-phase shift is copied by a node of the other colour.
It can also be moved past any phase shift of the other colour, swapping the two bits of that phase label in the process.
\begin{center}
 \input{tikz_files/11-copy2.0.tikz} $\qquad\qquad\qquad$ \input{tikz_files/11-comm2.0.tikz}
\end{center}
\vskip 0.2cm

\textbf{Colour change rule and Euler decomposition of \HadSpek{}}: The \HadSpek{} node swaps the colour of red and green nodes when it is applied to each input and output.
Furthermore, \HadSpek{} can be replaced by three green and red nodes of alternating colours, each with phase $01$.
This rule is called Euler decomposition in analogy to the Euler decomposition of general rotations in three-dimensional space into three rotations about two distinct axes.
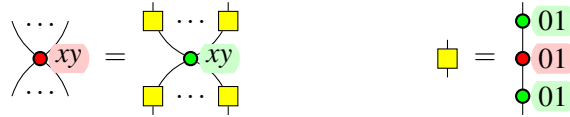
\begin{center}
 \input{tikz_files/colour_Spek.tikz} $\qquad\qquad\qquad$
 \input{tikz_files/Euler_Spek.tikz}
\end{center}
\vskip 0.2cm

Whenever a rule holds for any number of edges, that number may be zero.
For example, the colour change rule with zero inputs and zero outputs implies that \scalar{rn, label={[rphase]right:$xy$}} = \scalar{gn, label={[gphase]right:$xy$}} for any $x,y\in\{0,1\}$.

There are also two meta rules, rules that are not specified in terms of diagram equalities.
The first of these is that ``only the topology matters''.
This means that two diagrams represent the same process whenever they contain the same set of nodes connected up in the same ways, no matter how those nodes are arranged on the plane.
Secondly, we ``ignore all scalars that do not represent the empty relation'', meaning scalar subdiagrams can simply be dropped as long as they do not represent the empty relation.

\subsection{Universality of the graphical calculus}

The graphical calculus for the toy theory as defined in the previous two subsections is universal for the maximal knowledge fragment of Spekkens' toy bit theory with post-selected measurements.
This follows from the category-theoretical formulation of the toy theory in \cite{coecke_spekkens_2011}, where it is shown that all processes in the toy theory arise -- via parallel and sequential composition, and taking the relational converse -- from the 24 reversible transformations of a single toy bit together with a specific map called $\delta$ from one toy bit to two toy bits, and a post-selected measurement outcome denoted $\epsilon$.
It is straightforward to see that \HadSpek{} and the phase shifts suffice to construct all 24 reversible single-toy bit transformations, which correspond to the 24 permutations of the ontic states.
The maps $\delta$ and $\epsilon$ from \cite{coecke_spekkens_2011} are exactly the maps denoted by \splitnode{gn} and \effect{gn} here.
The graphical calculus allows parallel and sequential composition, as well as the taking of relational converses, which corresponds to flipping diagrams upside-down.
Therefore any process in the maximal knowledge fragment of Spekkens' toy bit theory with post-selected measurements can be represented graphically.

\subsection{Soundness of the graphical calculus}

Most of the rewrite rules of the toy theory graphical calculus can straightforwardly be checked to be sound by translating the diagrams on both sides of the equality into the corresponding maps.
For example, the left-hand side of the copy rule is, by definition, equal to:
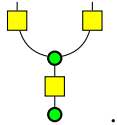
\begin{center}
 \input{tikz_files/copy_LHS.tikz}.
\end{center}
Using the definitions of \state{gn}, \HadSpek{}, and \splitnode{gn}, this diagram can be translated to the state:
\begin{equation}
 \{(1,1),(1,3),(3,1),(3,3)\}.
\end{equation}
The right-hand side of the copy rule translates to:
\begin{equation}\label{eq:copy_RHS}
 \{1,3\}\times\{1,3\}.
\end{equation}
By explicitly constructing the Cartesian product in \eqref{eq:copy_RHS}, the two relations are seen to be the same.
Therefore the copy rule is sound.

The other rewrite rules with fixed numbers of inputs and outputs -- i.e.\ the identity rule, bialgebra rule, 11-commutation rule, and Euler decomposition rule -- can be checked in the same way.
The colour change rule is sound by definition.
Soundness of both the 11-copy rule and the loop rule can be verified by induction over the number of outputs and/or inputs.

For soundness of the spider rule, we again rely on results from the categorical formulation of Spek\-kens' toy bit theory.
As shown in \cite{coecke_phase_2011}, the maps \splitnode{gn} and \effect{gn} form a category-theoretical \emph{observable}.
This means that any connected diagram constructed from these maps, their converses, wire crossings, and curved wires is determined solely by its number of inputs and outputs.
Graphically, this corresponds exactly to the spider law without phase labels \cite{coecke_povms_2008}.
The states \state{gn, label={[gphase]right:$xy$}} form a \emph{phase group} for this observable \cite{coecke_phase_2011}.
In particular, they form a group under the operation given by composition with \joinnode{gn}:
\begin{equation}
 \input{tikz_files/phase_group_operation.tikz},
\end{equation}
with group identity \state{gn} and all group elements being self-inverse.
From this, it follows that the spider law with phase labels is also sound.

Soundness of the topology meta-rule also follows from the category-theoretical formulation of the toy theory: The toy theory is modelled as a dagger compact closed category, therefore the results given in section \ref{s:graphical_rigorous} apply.

The ignore-scalars rule is also sound by the category-theoretical formulation.
There are exactly two relations from $I$ to $I$, the identity relation $\{(\bullet,\bullet)\}$ and the empty relation $\emptyset$.
Composing any relation with $\{(\bullet,\bullet)\}$ does not change the relation.
Thus dropping scalar subdiagrams is justified, as long as they are not the empty relation.

\subsection{The toy theory graphical calculus and the \ZX-calculus}

The toy theory graphical calculus is modelled after the \ZX-calculus for pure state stabilizer quantum mechanics with post-selected measurements, which also consists of green and red phased spiders and yellow nodes that change the colour of spiders \cite{coecke_interacting_2008,coecke_interacting_2011}.

The \ZX-calculus arises from the categorical formulation of quantum mechanics, and as explained in the previous section, the toy theory graphical calculus is closely related to the categorical formulation of Spekkens' toy bit theory.
Category-theoretically, the only difference between the toy theory and stabilizer quantum theory is the phase group of the respective observables: for the toy theory the phase group is isomorphic to the Klein Four group $\ZZ_2\times\ZZ_2$, whereas for stabilizer quantum theory the phase group is isomorphic to the cyclic group of order 4, $\ZZ_4$ \cite{coecke_phase_2011}.

Correspondingly, the rewrite rules of the toy theory graphical calculus that do not involve specific phases are exactly the same as those of the \ZX-calculus if the phase groups are swapped out.
In the stabilizer \ZX-calculus, the phase group is generally denoted by the angles $\{-\pi/2,0,\pi/2,\pi\}$ under addition modulo $2\pi$.
The Euler decomposition rule and 11-copy rules also have straightforward analogues in the \ZX-calculus, with $\pi/2$ phase shifts appearing in the Euler decomposition and $\pi$ taking the role of $11$.

The \ZX-calculus analogue to the $11$-commutation rule is the $\pi$-commutation rule:
\begin{equation}
 \input{tikz_files/pi-comm2.0.tikz}.
\end{equation}
At first glance this looks different to the $11$-commutation rule: the $\pi$-commutation rule sends any phase shift to its inverse whereas the $11$-commutation rule swaps the two bits denoting the phase.
In fact, both commutation rules can be expressed in the same way nevertheless.
Let $\varphi$ denote $11$ or $\pi$ and let $\theta$ be an arbitrary phase label for the respective theory.
We can write the two commutation rules in general form as:
\begin{equation}
 \input{tikz_files/general_comm.tikz},
\end{equation}
where $f$ is some map from the phase group to itself.
Then in both the \ZX-calculus for stabilizer quantum mechanics and in the graphical calculus for Spekkens' toy bit theory, the map $f$ can be characterised as follows: $f$ maps both $\varphi$ and the identity of the phase group back to themselves, but it swaps the remaining two elements of the phase group.

\section{Completeness of the graphical calculus}
\label{s:graphical_calculus}

We now show that the toy theory graphical calculus is complete by adapting the completeness proof for the stabilizer \ZX-calculus \cite{backens_zx-calculus_2013}.
There are several parts to the argument: First, we show that the results characterising all true equalities between stabilizer states from \cite{van_den_nest_graphical_2004}, which are central to the \ZX-calculus completeness proof, also hold in Spekkens' toy theory.
We then argue that it is sufficient to consider equalities between toy states rather than more general processes in the toy theory, because the toy theory has map-state duality.
Next, we prove that diagrams in the toy theory graphical calculus can be brought into a normal form called GS-LO form.
Finally, we show that the rewriting strategies used in the \ZX-calculus completeness proof also work in the toy theory graphical calculus.

Where the steps in the completeness argument for the toy theory differ only marginally from the corresponding steps in the stabilizer \ZX-calculus completeness proof, the proofs are left out or given in sketch form.
Longer proofs can be found in the appendices.

\subsection{The binary stabilizer formalism and the graph state theorems}
\label{s:binary}

Completeness of a graphical language means that any equality that can be derived in the standard formalism for the same theory can also be derived graphically.
Thus it is useful to have some simple way of characterising the equalities that can be derived in the underlying theory.

The completeness proof for the stabilizer \ZX-calculus makes use of two theorems about relationships between stabilizer states under local Clifford unitaries, i.e.\ unitary stabilizer operations that are tensor products of single-qubit unitaries.
Central to these results are graph states -- a class of stabilizer states whose entanglement structure is that of a finite simple graph, i.e.\ a graph with finitely many vertices (corresponding to the qubits), at most one edge between each pair of vertices (corresponding to the entanglement), and no self-loops.
Graph states on $n$ qubits can be represented in the binary formalism as follows \cite{van_den_nest_graphical_2004}:
\begin{equation}\label{eq:graph_state_matrix}
 S = \begin{pmatrix} \theta \\ I \end{pmatrix},
\end{equation}
where $\theta$ is a $n$ by $n$ symmetric matrix with zeroes along the diagonal and $I$ is the $n$ by $n$ identity matrix.
The matrix $\theta$ is in fact the adjacency matrix of the underlying graph: a 1 in position $(p,q)$ means that the $p$-th and $q$-th vertices are connected by an edge, a 0 means they are not connected.
The following results can be proved using the binary formalism:
\begin{theorem}[\cite{van_den_nest_graphical_2004}]\label{thm:van_den_nest1}
 Any stabilizer state can be transformed into some graph state by application of a local Clifford operation.
\end{theorem}
\begin{theorem}[\cite{van_den_nest_graphical_2004}]\label{thm:van_den_nest2}
 Two stabilizer states are equivalent under local Clifford operations if and only if the underlying graphs are related by a sequence of local complementations.
\end{theorem}

By local complementation we mean the following operation on a graph.

\begin{definition}
 Let $G=(V,E)$ be a graph with set of vertices $V$ and set of edges $E$. The \emph{local complementation about the vertex $v$} is the operation that inverts the subgraph generated by the neighbourhood of $v$ (but not including $v$ itself). Formally, a local complementation about $v\in V$ sends $G$ to the graph:
 \begin{equation}\label{eq:graph_local_complementation}
  G\star v = \left(V,E \triangle \big\{\{b,c\}\big|\{b,v\},\{c,v\}\in E\wedge b\neq c\big\}\right),
 \end{equation}
 where $\triangle$ denotes the symmetric set difference, i.e.\ $A\triangle B$ contains all elements that are contained either in $A$ or in $B$ but not in both.
\end{definition}

We now show that these results translate to the toy theory.

As described in section \ref{s:spekkens}, a state of maximal knowledge on $n$ toy bits is given by a set of $n$ commuting quadrature variables, together with the values for each of the variables.
These quadrature variables can be represented as binary vectors, similar to the representation of Pauli products, where the $m$-th and $(m+n)$-th component together encode the quadrature variable acting on the $m$-th toy bit according to the following encoding:
\begin{align}
 X &\mapsto 01, \\
 Z &\mapsto 10, \text{ and} \\
 X\oplus Z &\mapsto 11,
\end{align}
with 00 indicating that no quadrature variable is acting on the given toy bit.
Thus, ignoring the values of the quadrature variables, any state of maximal knowledge can be described by a binary $2n$ by $n$ matrix in the same way as a pure quantum state.

\begin{lemma}
 A binary $2n$ by $n$ matrix $S$ represents a valid state in the toy theory if and only if $S^T J S=0$, where:
 \begin{equation}
  J = \begin{pmatrix}0&I\\I&0\end{pmatrix}
 \end{equation}
 with $I$ the $n$ by $n$ identity matrix.
\end{lemma}
This follows from the principle of classical complementarity, as shown in \cite{spekkens_quasi-quantization_2014}.
In the binary picture, the valid reversible transformations of the toy theory are represented by $2n$ by $2n$ binary matrices $Q$ satisfying $Q^T J Q = J$.

The conditions for $2n$ by $n$ binary matrices to represent valid states and the condition for $2n$ by $2n$ binary matrices to represent valid transformations are exactly the same as in the binary formalism for stabilizer quantum mechanics.
Therefore the binary matrix formalism for Spekkens' toy bit theory is exactly the same as the check matrix formalism for stabilizer quantum mechanics, if one ignores the values of the quadrature variables in the former and the eigenvalues in the latter.
An equivalent result was shown in \cite{pusey_stabilizer_2012}, albeit not using check matrices.

When considering graph states and local Clifford transformations in quantum mechanics, it is reasonable to ignore the eigenvalues in the stabilizer formalism because the eigenvalue for each stabilizer of a graph state can be changed by a local Clifford transformation that keeps all the other properties of the state invariant.
Define graph states in Spekkens' toy bit theory to be states having the same check matrix representation as some graph state in stabilizer quantum mechanics.
Then the values of the quadrature variables can be ignored for the same reason as eigenvalues in quantum theory.

Thus theorems \ref{thm:van_den_nest1} and \ref{thm:van_den_nest2} carry over to the toy theory, i.e.\ we have:

\begin{theorem}\label{thm:vdN1_toy}
 Any toy stabilizer state is equivalent to some toy graph state under local toy transformations $\sigma\in (S_4)^n$.
\end{theorem}

\begin{theorem}\label{thm:vdN2_toy}
 Two toy graph states on the same number of toy bits are equivalent under local toy transformations if and only if there is a sequence of local complementations that transform one graph into other.
\end{theorem}

Local complementations on a graph are defined just as above.

\subsection{Map-state duality for the toy theory}

The graph state theorems, as implied by the name, apply only to toy states, not more general processes in the toy theory.
Yet it suffices to consider only equalities between states in order to get a completeness result for the entire theory.
This is because the toy theory exhibits map-state duality: as in quantum theory, there exists a bijection between operators from $n$ to $m$ toy bits and states on $(n+m)$ toy bits.
This duality is also known as the Choi-Jamio{\l}kowski isomorphism.
Denoting the map by $A$ and its corresponding state by $B$, the isomorphism can be represented diagrammatically as follows:
\begin{equation}\label{eq:Choi-Jamiolkowski}
 \input{tikz_files/Choi-Jamiolkowski1.tikz} \quad\Longleftrightarrow\quad \input{tikz_files/Choi-Jamiolkowski2.tikz}
\end{equation}
This result follows directly from the snake equations given in Fig. \ref{fig:compact_structure}, which hold in the toy theory graphical calculus via the topology rule.

The Choi-Jamio{\l}kowski isomorphism allows toy theory operators to be turned into states.
Any equalities derived between these states then apply also to the original operators.
Thus, a completeness result for the entire toy theory can be derived by considering only toy states.

\subsection{Graph states and related diagrams in the toy theory graphical calculus}
\label{s:toy_diagrams}

In the first part of this section, we defined graph states for the toy theory via their check matrices.
We now show that they also have an elegant graphical representation.

\begin{definition}\label{dfn:graph_state}
 Let $G$ be a finite simple undirected graph, i.e.\ a graph with finitely many vertices, at most one edge between any pair of vertices, and no self-loops.
 Let the set of vertices be $V$ and the set of edges $E$.
 The associated graph state in the toy theory graphical calculus comprises the following:
 \begin{mitem}
  \item for each vertex in $V$, a green node with one output, and
  \item for each edge in $E$, a copy of \HadSpek{} connected to the green nodes representing the vertices at either end of the edge.
 \end{mitem}
\end{definition}

To show that this definition is equivalent to the previous one, we consider the operators stabilizing the graph state.

\begin{lemma}
\label{lem:fixpoint}
 Let
 \begin{tikzpicture}[baseline=-0.1 cm]
	\begin{pgfonlayer}{nodelayer}
		\node [style=none] (0) at (-0.75, 0.3) {};
		\node [style=none] (1) at (0.75, 0.3) {};
		\node [style=none] (2) at (-0.75, 0) {};
		\node [ellipse, fill=White, draw=Black, minimum width=2 cm, style=none] (3) at (0, -0) {$G$};
		\node [style=none] (4) at (0.75, 0) {};
		\node [style=none] (5) at (0, 0.27) {$\ldots$};
	\end{pgfonlayer}
	\begin{pgfonlayer}{edgelayer}
		\draw (0.center) to (2.center);
		\draw (1.center) to (4.center);
	\end{pgfonlayer}
 \end{tikzpicture}
 denote the state associated with a graph $G$. Then for any vertex $v\in V$,
 \begin{center}
  \input{tikz_files/eigenstate.tikz}
 \end{center}
 where $\alpha_k=11$ if $\{v,k\}\in E$ and $\alpha_k=00$ otherwise.
 This means that
 \begin{tikzpicture}[baseline=-0.1 cm]
	\begin{pgfonlayer}{nodelayer}
		\node [style=none] (0) at (-0.75, 0.3) {};
		\node [style=none] (1) at (0.75, 0.3) {};
		\node [style=none] (2) at (-0.75, 0) {};
		\node [ellipse, fill=White, draw=Black, minimum width=2 cm, style=none] (3) at (0, -0) {$G$};
		\node [style=none] (4) at (0.75, 0) {};
		\node [style=none] (5) at (0, 0.27) {$\ldots$};
	\end{pgfonlayer}
	\begin{pgfonlayer}{edgelayer}
		\draw (0.center) to (2.center);
		\draw (1.center) to (4.center);
	\end{pgfonlayer}
 \end{tikzpicture}
 is an eigenstate of any operator that applies \phase{rn,label={[rphase]right:$11$}} to one of the vertices and \phase{gn,label={[gphase]right:$11$}} to all neighbours of that vertex.
\end{lemma}

This lemma follows from the rules of the red-green calculus for the toy theory; the proof is entirely analogous to that for the \ZX-calculus, where $\pi$-phase shifts appear instead of the 11-phase shifts \cite{duncan_graph_2009}.

\begin{corollary}
 Definition \ref{dfn:graph_state} is equivalent to the definition of toy theory graph states via their check matrices.
\end{corollary}
\begin{proof}
 Lemma \ref{lem:fixpoint} gives $n$ quadrature variables of which the diagram is an eigenstate: a phase shift \phase{rn, label={[rphase]right:$11$}} on the $p$-th output corresponds to a term $X_p$ in the quadrature variable, a phase shift \phase{gn, label={[gphase]right:$11$}} on the $q$-th output corresponds to a term $Z_q$.
 Translate each of these variables into a binary vector as described in section \ref{s:binary}, and assemble the resulting vectors as the columns of a matrix with the vector for the variable involving the term $X_m$ as the $m$-th column.
 The resulting check matrix then has the form given in \eqref{eq:graph_state_matrix}.
\end{proof}

The local complementation operations from theorem \ref{thm:vdN2_toy} can be derived from the rules of the toy theory graphical calculus.
Note that from now on we use the term ``local complementation'' in a slightly different sense as before: previously, the term referred to an operation on graphs only.
From now on, we use the same term to refer to an operation on graph states together with the application of a local operation to all the toy bits that keeps the overall toy state invariant. 

\begin{lemma}
\label{lem:local_complementation}
 The following local complementation rule holds in the red-green calculus for the toy theory:
 \begin{center}
  \input{tikz_files/local_complementation_Spek1.tikz}
 \end{center}
 where $\alpha_k=01$ if $\{v,k\}\in E$ and $\alpha_k=00$ otherwise, and $G\star v$ denotes the graph-theoretical local complementation as defined in \eqref{eq:graph_local_complementation}.
\end{lemma}

A sketch proof of this lemma can be found in appendix \ref{s:appendix_graph_states}.

A \emph{local complementation along an edge} $\{v,w\}\in E$ consists of a local complementation about $v$, a local complementation about $w$, and another local complementation about $v$, yielding:
 
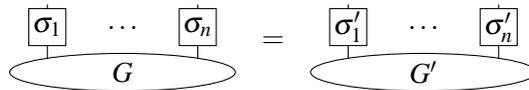
\begin{center}
  \begin{tikzpicture}[baseline=-0.1cm]
	\begin{pgfonlayer}{nodelayer}
		\node [style=none] (0) at (1, 0.5) {};
		\node [style=none] (1) at (1, -0.5) {};
		\node [ellipse, fill=White, draw=Black, minimum height=0.5 cm, minimum width=3 cm, style=none] (2) at (2, -0.4) {$G'$};
		\node [style=none] (3) at (0, -0) {$=$};
		\node [style=none] (4) at (3, -0.5) {};
		\node [style=none] (5) at (3, 0.5) {};
		\node [ellipse, fill=White, draw=Black, minimum height=0.5 cm, minimum width=3 cm, style=none] (6) at (-2, -0.4) {$G$};
		\node [style=none] (7) at (-3, 0.5) {};
		\node [style=none] (8) at (-3, -0.5) {};
		\node [style=none] (9) at (-1, 0.5) {};
		\node [style=none] (10) at (-1, -0.5) {};
		\node [style=none] (11) at (-2, 0.2) {$\ldots$};
		\node [rectangle,fill=white,draw=black,minimum height=14pt,minimum width=14pt,inner sep=0pt] (12) at (-3, 0.2) {$\sigma_1$};
		\node [rectangle,fill=white,draw=black,minimum height=14pt,minimum width=14pt,inner sep=0pt] (13) at (-1, 0.2) {$\sigma_n$};
		\node [rectangle,fill=white,draw=black,minimum height=14pt,minimum width=14pt,inner sep=0pt] (14) at (1, 0.2) {$\sigma_1'$};
		\node [rectangle,fill=white,draw=black,minimum height=14pt,minimum width=14pt,inner sep=0pt] (15) at (3, 0.2) {$\sigma_n'$};
		\node [style=none] (16) at (2, 0.2) {$\ldots$};
	\end{pgfonlayer}
	\begin{pgfonlayer}{edgelayer}
		\draw (0.center) to (1.center);
		\draw (5.center) to (4.center);
		\draw (7.center) to (8.center);
		\draw (9.center) to (10.center);
	\end{pgfonlayer}
  \end{tikzpicture}
 \end{center}
where:
 \[
  \sigma_j' = \begin{cases} \sigma_j \circ (23) & \text{if }j\in\{v,w\} \\ \sigma_j & \text{otherwise} \end{cases}
 \]
and $G' = (V,E')$ satisfies the following properties:
 \begin{itemize}
  \item $G' = ((G\star v)\star w)\star v = ((G\star w)\star v)\star w$;
  \item $\{v,w\}\in E'$;
  \item for $j\in V\setminus\{v,w\}$, $\{j,v\}\in E'\Leftrightarrow \{j,w\}\in E$ and $\{j,w\}\in E'\Leftrightarrow \{j,v\}\in E$, i.e.\ a vertex $j$ is adjacent to $v$ in $G'$ if and only if $j$ was adjacent to $w$ in $G$ and correspondingly with $v$ and $w$ exchanged;
  \item for $p,q\in V\setminus\{v,w\}$, let $P$ be the intersection of $p$'s neighbourhood with $\{v,w\}$, i.e.\ $v\in P$ if $\{p,v\}\in E$ and $w\in P$ if $\{p,w\}\in E$, and define $Q$ correspondingly. Then the edge $\{p,q\}$ is toggled if and only if $P,Q$ and $\emptyset$ are pairwise distinct.
 \end{itemize}
This operation is symmetric under interchange of the two vertices $v$ and $w$.

It will be useful to have a normal form for reversible single-toy bit operators.
\begin{lemma}\label{lem:single_toy_bit_operator}
 Any reversible single-toy bit operator, i.e.\ any diagram or subdiagram consisting solely of phase shifts and \HadSpek, can be written uniquely in one of the following forms:
 \begin{center}
  \input{tikz_files/single_toybit.tikz}
 \end{center}
 where $a,b,c,d,e,f,g\in\{0,1\}$ and $\bar{e}=e\oplus 1$, with $\oplus$ denoting addition modulo 2.
\end{lemma}
This is straightforward to check, analogously to the corresponding result in the \ZX-calculus.
In the following, whenever we talk about reversible single-toy bit operators we will assume that they are normalised as in the above lemma.

\begin{definition}
 A diagram in the red-green calculus for Spekkens' toy theory is called a \emph{GS-LO diagram} (graph state with local operators) if it consists of a graph state as in definition \ref{dfn:graph_state} with reversible single-toy bit operators on each output.
\end{definition}

GS-LO diagrams play a central role in the graphical calculus for the toy theory, as shown by the following theorem.

\begin{theorem}\label{thm:GS_LO}
 Any state diagram in the red-green calculus for Spekkens's toy theory is equal to some GS-LO diagram according to the rewrite rules.
\end{theorem}

This theorem is analogous to theorem 7 in \cite{backens_zx-calculus_2013}. A sketch of the parts of the proof that differ from the \ZX-calculus case can be found in appendix \ref{s:appendix_graph_states}.
The proof relies on the fact that diagrams can be decomposed into reversible single-toy bit operators and four basic spiders:
\[
 \joinnode{gn},\quad \splitnode{gn},\quad \state{gn},\quad \text{and}\quad \effect{gn}.
\]
Local complementations together with fixpoint operations can be used to change the reversible single-toy bit operator on any given toy bit in the graph state to any desired operator, similar to the equivalent result for the \ZX-calculus \cite{backens_zx-calculus_2013}.
Thus it can be shown that any diagram consisting of a basic spider composed with a GS-LO diagram can be rewritten into GS-LO form.
The only basic spider with no inputs is \state{gn}, so any state diagram must contain at least one copy of that; furthermore \state{gn} is a GS-LO diagram.
Therefore, by induction, any diagram can be brought into GS-LO form.

The GS-LO form is not unique, i.e.\ there may be different GS-LO diagrams representing the same state.
It is not clear how to define a unique normal form, but it is possible to reduce the number of diagrams needing to be considered further.

\begin{definition}\label{dfn:rGS-LO}
 A diagram in Spekkens's toy theory is said to be in \emph{reduced GS-LO (or rGS-LO) form} if it is in GS-LO form and satisfies the following additional conditions:
 \begin{itemize}
  \item All the reversible single-toy bit operators belong to the set:
   \begin{equation}\label{eq:reduced_vertex_operators}
    R = \input{tikz_files/set_R.tikz}
   \end{equation}
  \item Two adjacent vertices must not both have vertex operators that include red nodes.
 \end{itemize}
\end{definition}

\begin{theorem}\label{thm:rGS-LO}
 Any toy stabilizer state diagram is equal to some rGS-LO diagram within the graphical calculus.
\end{theorem}

A sketch proof of this theorem can be found in appendix \ref{s:appendix_graph_states}.

Reduced GS-LO diagrams are still not unique, as shown by the following two propositions, sketch proofs of which can be found in the appendix.

\begin{proposition}\label{prop:rGS-LO_transformation1}
 Suppose a rGS-LO diagram contains a pair of neighbouring toy bits $p$ and $q$ in the following configuration, where $a,b\in\{0,1\}$:
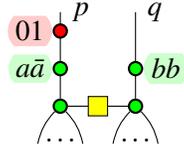
 \begin{center}
  \input{tikz_files/rGS-LO_transformation1.tikz}
 \end{center}
 Then a local complementation about $q$, followed by a local complementation about $p$, yields a diagram which can be brought into rGS-LO form by at most two applications of the fixpoint rule.
\end{proposition}

\begin{proposition}\label{prop:rGS-LO_transformation2}
 Suppose a rGS-LO diagram contains a pair of neighbouring toy bits $p$ and $q$ in the following configuration, where $a,b\in\{0,1\}$:
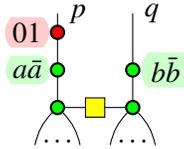
 \begin{center}
  \input{tikz_files/rGS-LO_transformation2.tikz}
 \end{center}
 Then a local complementation along the edge $\{p,q\}$ yields a diagram which can be brought into rGS-LO form by at most two applications of the fixpoint rule.
\end{proposition}

With the definitions and results in this section, state diagrams in the toy theory graphical calculus can be simplified significantly.
By map-state duality, the results can be applied to arbitrary diagrams.

For completeness it remains to be shown that whenever two rGS-LO diagrams represent the same toy state, they can be rewritten into each other using the rewrite rules for the toy theory graphical calculus.

\subsection{Equalities between rGS-LO diagrams}

The graphical calculus is complete for toy theory states if, given any two rGS-LO diagrams representing the same state, we can show that they are equal using the rules of the graphical calculus.
In this section, we exhibit an algorithm for rewriting two diagrams representing the same toy state to be identical.
As rewrite rules are invertible, this is equivalent to being able to rewrite one diagram into the other.
The algorithm is adapted from a similar one for the stabilizer \ZX-calculus \cite{backens_zx-calculus_2013}.

Given two toy state diagrams on the same number of toy bits, we start by pairing up red nodes between the two diagrams.

\begin{definition}
 A pair of rGS-LO diagrams on the same number of toy bits is called \emph{simplified} if there are no pairs of toy bits $p,q$ such that $p$ has a red node in its vertex operator in the first diagram but not in the second, $q$ has a red node in the second diagram but not in the first, and $p$ and $q$ are adjacent in at least one of the diagrams.
\end{definition}

\begin{proposition}
 Any pair of rGS-LO diagrams on $n$ toy bits can be simplified.
\end{proposition}

The proof of this proposition is analogous to the stabilizer \ZX-calculus case in \cite{backens_zx-calculus_2013}.
The idea is to use the equivalence operations of rGS-LO diagrams given in propositions \ref{prop:rGS-LO_transformation1} and \ref{prop:rGS-LO_transformation2} to shift red nodes to different toy bits in the diagram.

Simplifying a pair of diagrams is not an arbitrary process: for a simplified pair of diagrams, it is straightforward to decide whether or not they are equal according to the rewrite rules.
Firstly, if there exist red nodes that cannot be paired up between the two diagrams, then the diagrams cannot represent the same state, as shown by the following lemma.

\begin{lemma}\label{lem:unpaired_red}
 Consider a simplified pair of rGS-LO diagrams and suppose there exists an unpaired red node, i.e.\ there is a toy bit $p$ which has a red node in its vertex operator in one of the diagrams, but not in the other. Then the two diagrams are not equal.
\end{lemma}

This lemma is proved in appendix \ref{s:completeness}.

The existence of unpaired red nodes is not the only sign that a simplified pair of diagrams cannot be equal.
In fact, a simplified pair of diagrams are either identical or they do not represent the same state.

\begin{theorem}\label{thm:state_completeness}
 The two diagrams making up a simplified pair of rGS-LO diagram are equal, i.e.\ they correspond to the same toy theory state, if and only if they are identical.
\end{theorem}

The proof of this theorem is analogous to that of theorem 18 in \cite{backens_zx-calculus_2013}.

By map-state duality for the toy theory, as given in \eqref{eq:Choi-Jamiolkowski}, and invertibility of the rewrite rules, theorem \ref{thm:state_completeness} directly implies:

\begin{theorem}
 The red-green calculus is complete for Spekkens' toy bit theory.
\end{theorem}

Equalities between two diagrams in the toy theory graphical calculus can be derived as follows: if the diagrams are not states, bend all inputs around to become outputs.
Bring the two diagrams into GS-LO form and thus into rGS-LO form.
Simplify the pair of diagrams.
Then either the two diagrams are identical, in which case some of the rewrite steps can be inverted to get a sequence of rewrites transforming one diagram into the other, or they are not identical, in which case the two diagrams do not represent the same operator, so there is no equality to derive.

If the diagrams were not states to begin with, the appropriate outputs can be bent back into inputs in all diagram.
This yields a sequence of valid rewrites transforming one of the original diagrams into the other.

\section{Conclusions}
\label{s:conclusions}

We have defined a graphical calculus for Spekkens' toy bit theory and shown that it is universal, sound, and complete for the maximal knowledge fragment of the theory with post-selected measurements.
This means that the graphical calculus has the full power of any formalism for analysing the toy theory.
Our graphical calculus is modelled after the \ZX-calculus, a similar universal, sound, and complete graphical calculus for pure state qubit stabilizer quantum mechanics with post-selected measurements.
Therefore similarities and differences between stabilizer quantum mechanics and the toy bit theory can be analysed entirely graphically.

A potential next step for this research programme is to implement the rewrite rules and algorithms involved in the completeness proof in the software system \emph{Quantomatic}, which enables automated and semi-automated manipulation of diagrams in the \ZX-calculus and similar graphical languages \cite{quantomatic}.
That way, diagrams can be simplified and equalities between toy theory diagrams can be derived automatically.
If the corresponding algorithms for the \ZX-calculus are implemented as well, Quantomatic can compare the two theories automatically.

So far, we have only considered pure state qubit stabilizer quantum mechanics and the maximal knowledge fragment of Spekkens' toy bit theory.
An obvious next step would be to extend the graphical calculi to mixed states in the quantum case and states of less-than-maximal knowledge in the toy theory.
The category-theoretical formulations underlying the graphical calculi can easily be extended in this way using the CPM-construction \cite{selinger_dagger_2007} and these extensions carry over to categorical graphical calculi.

Furthermore, it would be interesting to extend this argument to stabilizer quantum mechanics for higher dimensional systems and the higher-dimensional toy theory.
Some steps in this direction have been made by generalising the \ZX-calculus to qudits and Spekkens' toy theory for systems of dimension greater than two, though it is still unclear whether these graphical languages are complete \cite{ranchin_depicting_2014}.

Rigorous graphical languages have many applications in the analysis of quantum physics and related theories.

\section*{Acknowledgements}

The authors would like to thank Dominic Horsman and Matt Pusey for comments on drafts of this paper.
MB acknowledges financial support from the EPSRC.

\bibliographystyle{eptcs}
\bibliography{graphical_calculus_for_Spekkens}

\newpage
\appendix
\section{Appendix}

\subsection{Proofs of results about graph states, GS-LO diagrams, and rGS-LO diagrams}
\label{s:appendix_graph_states}

Here, we give the proofs for results stated in section \ref{s:toy_diagrams} where they differ significantly from the corresponding proofs in the \ZX-calculus.
The \ZX-calculus is introduced in \cite{coecke_interacting_2011} and extended in \cite{duncan_graph_2009}.
The completeness proof for the stabilizer \ZX-calculus can be found in \cite{backens_zx-calculus_2013}.

\begin{proof}[Lemma \ref{lem:local_complementation}, sketch]
 The proof is similar to the \ZX-calculus case as given by Duncan and Perdrix \cite{duncan_graph_2009}.
 We show here as an example the case of the complete graph on three vertices (rearranged with two inputs at the bottom for ease of reading):
 \begin{center}
  \input{tikz_files/local_complementation_derivation1.tikz}
 \end{center}
 \begin{center}
  \input{tikz_files/local_complementation_derivation2.tikz}
 \end{center}
 The first equality uses the decomposition of \HadSpek{} in terms of red and green phase shifts
 In the second step, the spider rule is used to merge the green phases with their green neighbours.
 Subsequently, the red and green phased spiders are ``pulled apart'', again using the spider law.
 In the fifth step, the colour change law and the fact that \HadSpek{} is self-inverse are used to change the green node at the top into a red one.
 The next step is an application of the bialgebra law.
 The penultimate step uses the fact that \state{rn,label={[rphase]right:$01$}} = \state{gn,label={[gphase]right:$01$}}, which is the case $a=0$ of \eqref{eq:red_green_states} below.
 Lastly, the colour change rule is applied again.

 The full proof then proceeds by induction over the number of vertices in the graph state.
\end{proof}

\begin{proof}[Theorem \ref{thm:GS_LO}, sketch]
 The proof is analogous to the proof of Theorem 7 in \cite{backens_zx-calculus_2013}, noting the following facts:
 \begin{itemize}
  \item Let $a\in\{0,1\}$ and $\bar{a}=a\oplus 1$, then:
   \begin{equation} \label{eq:red_green_states}
    \input{tikz_files/single_toy_state.tikz}
   \end{equation}
   Here, the first step uses the fact that \HadSpek{} is self-inverse and the second step uses the decomposition of \HadSpek{} into red and green phase shifts.
   The third step is an application of the spider law to merge the bottom two nodes, which is again used in the fourth step to pull apart the green node.
   In the fifth step, the bottom red node is copied: this works for both values of $a$.
   The penultimate step, involves dropping the scalar diagram on the left and merging the two red nodes in the non-scalar part by the spider law.
   The last equality is by the colour change law.
  \item Any single toy bit operator can be written as
   \begin{center}
    \input{tikz_files/green-red-green.tikz}
   \end{center}
   for some $a,b,c,d,e,f\in\{0,1\}$.
  \item
   \begin{tikzpicture}[baseline=-.1cm]
	\begin{pgfonlayer}{nodelayer}
		\node [style={gn,label={[gphase]right:$11$}}] (0) at (0, 0) {};
	\end{pgfonlayer}
   \end{tikzpicture} and
   \begin{tikzpicture}[baseline=-.1cm]
	\begin{pgfonlayer}{nodelayer}
		\node [style={rn,label={[rphase]right:$11$}}] (0) at (0, 0) {};
	\end{pgfonlayer}
   \end{tikzpicture} denote the zero scalar.
  \item A loop with a \HadSpek\ node in it disappears:
   \begin{equation}
    \input{tikz_files/HadSpek_loop.tikz}
   \end{equation}
 \end{itemize}
\end{proof}

\begin{proof}[Theorem \ref{thm:rGS-LO}]
 By theorem \ref{thm:GS_LO}, any state diagram in the toy theory is equal to some GS-LO diagram.
 Lemma \ref{lem:single_toy_bit_operator} shows that each vertex operator in the GS-LO diagram can be brought into the form:
 \begin{center}
  \input{tikz_files/single_toybit.tikz},
 \end{center}
 where $a,b,c,d,e,f,g\in\{0,1\}$.
 Note that the cases $c=0=d$ and $f=0=g$ of the above normal forms correspond exactly to the elements of $R$ as defined in \eqref{eq:reduced_vertex_operators}.
 A local complementation about a vertex $v$ pre-multiplies the vertex operator of $v$ with \phase{rn,label={[rphase]right:$01$}} and a fixpoint operation with \phase{rn,label={[rphase]right:$11$}}, so any vertex operator can be brought into one of the above forms by some combination of local complementations and fixpoint operations about the corresponding vertex.
 The other effects of local complementations are to toggle some of the edges in the graph state and to pre-multiply the vertex operators of neighbouring vertices by \phase{gn,label={[gphase]right:$01$}}, whereas  fixpoint operations leave the edges invariant and pre-multiply the vertex operators of neighbouring vertices by \phase{gn,label={[gphase]right:$11$}}.
 The set $R$ is not mapped to itself under repeated pre-multiplication with \phase{gn,label={[gphase]right:$01$}}: this transformation sends the set $\{\phase{gn,label={[gphase]right:$ab$}}\}$ for $a,b\in\{0,1\}$ to itself, but it maps:
 \begin{equation}
  \input{tikz_files/R_pre_multiplied.tikz}
 \end{equation}
 The normal form of a vertex operator contains at most two red nodes.
 Once a vertex operator is in one of the forms in $R$, pre-multiplication by green phase operators does not change the number of red nodes it contains when expressed in normal form.
 Thus the process of removing red nodes from the vertex operators by applying local complementations must terminate after at most $2n$ steps for an $n$-toy bit diagram, at which point all vertex operators are elements of the set $R$.

 With all vertex operators in $R$, suppose there are two adjacent toy bits $u$ and $v$ which both have red nodes in their vertex operators, i.e.\ there is a subdiagram of the form:
 \begin{equation}\label{eq:neighbouring_red_nodes}
  \input{tikz_files/adjacent_red_nodes.tikz}
 \end{equation}
 with $a,b,\in\{0,1\}$.
 A local complementation along the edge $\{u,v\}$ maps the vertex operator of $u$ to:
 \begin{equation}\label{eq:pivoting}
  \input{tikz_files/vertex_operator_u.tikz}
 \end{equation}
 and similarly for $v$.
 After this, if $a=1$, we apply a fixpoint operation to $u$ and if $b=1$ we apply a fixpoint operation to $v$.
 After this, the vertex operators on both $u$ and $v$ are green phase operators.
 Vertex operators of toy bits adjacent to $u$ or $v$ are pre-multiplied with some power of \phase{gn,label={[gphase]right:$11$}}, which maps $R\to R$.
 Thus each such operation removes the red nodes from a pair of adjacent toy bits and leaves all vertex operators in the set $R$.
 Hence after at most $n/2$ such operations, it will be impossible to find a subdiagram as in \eqref{eq:neighbouring_red_nodes}.
 Thus, the diagram is in reduced GS-LO form.
\end{proof}

\begin{proof}[Proposition \ref{prop:rGS-LO_transformation1}, sketch]
 The effect of the local complementations on the vertex operators of $p$ and $q$ is as follows:
 \begin{equation}
  \input{tikz_files/rGS-LO_transformation1_proof.tikz}
 \end{equation}
 If $a=1$, we apply a fixpoint operation to $p$ and if $b=1$, we apply a fixpoint operation to $q$; then the vertex operators of $p$ and $q$ are in $R$.
 The fixpoint operations add \phase{gn,label={[gphase]right:$11$}} to neighbouring toy bits, which maps the set $R$ to itself.
 As fixpoint operations do not change any edges, we do not have to worry about them when considering whether the rest of the diagram satisfies definition \ref{dfn:rGS-LO}.

 The rest of the proof is analogous to the stabilizer QM case in \cite{backens_zx-calculus_2013}.
\end{proof}

\begin{proof}[Proposition \ref{prop:rGS-LO_transformation2}]
 After the local complementation along the edge, the vertex operator of $p$ is given by \eqref{eq:pivoting}.
 For the vertex operator of $q$, we have:
 \begin{equation}
  \input{tikz_files/rGS-LO_transformation2_proof.tikz}
 \end{equation}
 Thus if $a$ or $b$ is 1, we apply a fixpoint operator to the appropriate vertex.
 From the properties of local complementations along edges it follows that the overall transformation preserves the two properties of rGS-LO states.
\end{proof}

\subsection{Proof of completeness result}
\label{s:completeness}

The arguments in the following proof closely follow the proof of Lemma 17 in \cite{backens_zx-calculus_2013}.
As the diagrams are complicated and differ in subtle ways from the \ZX-calculus ones, the proof is nevertheless produced in full here.

\begin{proof}[Lemma \ref{lem:unpaired_red}]
 Let $D_1$ be the diagram in which $p$ has the red node, $D_2$ the other diagram. There are multiple cases:

 \emph{In either diagram, $p$ has no neighbours}: In this case, the overall state factorises and the two diagrams are equal only if the two states of $p$ are the same.
 But:
 \begin{equation}
  \input{tikz_files/distinct_single_states.tikz}
 \end{equation}
 for $a,b,c\in\{0,1\}$, so the diagrams must be unequal.

 \emph{$p$ is isolated in one of the diagrams but not in the other}: We argue in section \ref{s:binary} that, as in stabilizer QM, two toy graph states with local operators are equal only if one can be transformed into the other via a sequence of local complementations with corresponding changes to the local operators.
 As a local complementation never turns a vertex with neighbours into a vertex without neighbours, or conversely, the two diagrams cannot be equal.

 \emph{$p$ has neighbours in both diagrams}: Without loss of generality, assume that $p$ is the first toy bit.
 Let $N_1$ be the set of all toy bits that are adjacent to $p$ in $D_1$, and define $N_2$ similarly.
 The vertex operators of any toy bit in $N_1$ must be green phases in both diagrams.
 In $D_1$, this is because of the definition of rGS-LO diagrams, in $D_2$ it is because the pair of diagrams is simplified.
 Suppose the original diagrams involve $n$ toy bits each.
 Let $G$ be the graph on $n$ vertices (named according to the same convention as in $D_1$ and $D_2$) whose edges are $\{\{p,v\} | v\in N_1\}$.
 Now consider the following diagram:
 \begin{equation}\label{eq:isolate_p}
  \input{tikz_files/isolate_p.tikz}
 \end{equation}
 where the ellipse labelled $G$ denotes the toy graph state corresponding to $G$, except that each vertex in the graph has not only an output but also an input.
 Call this diagram $U$.
 It is easy to see that $U$ is invertible: composing it with itself upside-down yields the identity.
 Therefore composing this diagram with $D_1$ and $D_2$ will yield two new diagrams which are equal if and only if $D_1=D_2$.
 We will denote the new diagrams by $U\circ D_1$ and $U\circ D_2$ and show that, no matter what the properties of $D_1$ and $D_2$ are (beyond the existence of an unpaired red node on $p$),
 \begin{mitem}
  \item in $U\circ D_1$, the toy bit $p$ is in state \state{rn} or \state{rn,label={[rphase]right:$11$}};
  \item in $U\circ D_2$, $p$ is either entangled with other toy bits, or in one of the states \state{gn,label={[gphase]right:$ab$}}, where $a,b\in\{0,1\}$.
 \end{mitem}
 By the arguments used in the first two cases, this implies that $U\circ D_1\neq U\circ D_2$ and therefore $D_1\neq D_2$.

 Let $n=\abs{N_1}$, $m=\abs{N_1\cap N_2}$, and suppose the toy bits are arranged in such a way that the first $m$ elements of $N_1$ are those which are also elements of $N_2$, if there are any.
 Consider first the effect on diagram $D_1$.
 The local operator on $p$ combines with the single-toy bit operators from $U$ to:
 \begin{equation}
  \input{tikz_files/local_operator_on_p.tikz},
 \end{equation}
 where $a\in\{0,1\}$.
 As green phase shifts can be pushed through other green nodes, the subdiagram involving $p$ and the elements of $N_1$ in $U\circ D_1$ is equal to:
 \begin{equation}
  \input{tikz_files/U_after_D1.tikz}
 \end{equation}
 Here, $b_1,\ldots,b_n,c_1,\ldots,c_n\in\{0,1\}$.
 Note that at the end $p$ is isolated and in the state \state{rn,label={[rphase]right:$aa$}}.
 The fact that we have ignored all toy bits not originally adjacent to $p$ in $D_1$ does not change that.

 Next consider $U\circ D_2$.
 As $N_1$ is not in general equal to $N_2$, the subdiagram consisting of $p$ and vertices in $N_1$ looks as follows:
 \begin{center}
  \input{tikz_files/U_after_D2.tikz}
 \end{center}
 where $l=m+1$ and $d,e,f_1,\ldots,f_n,g_1,\ldots,g_n\in\{0,1\}$.
 Note that we neglect edges that do not involve $p$ and also edges between $p$ and vertices not in $N_1$.
 We will now distinguish different cases, depending on the values of $d$ and $e$.

 If $d=0,e=1$ apply a local complementation about $p$.
 This does not change the edges incident on $p$:
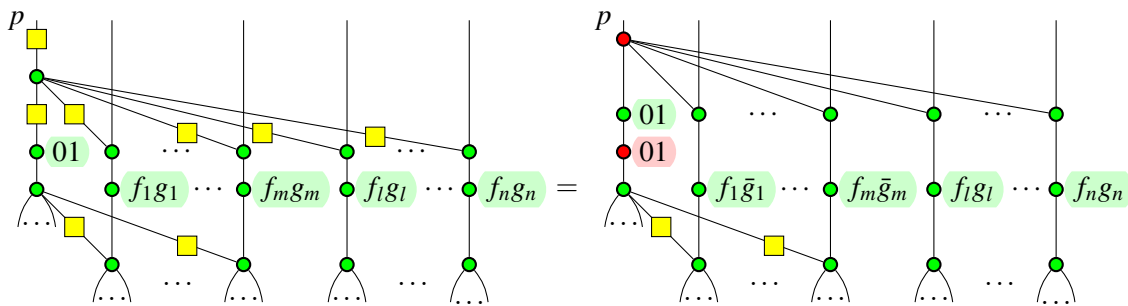
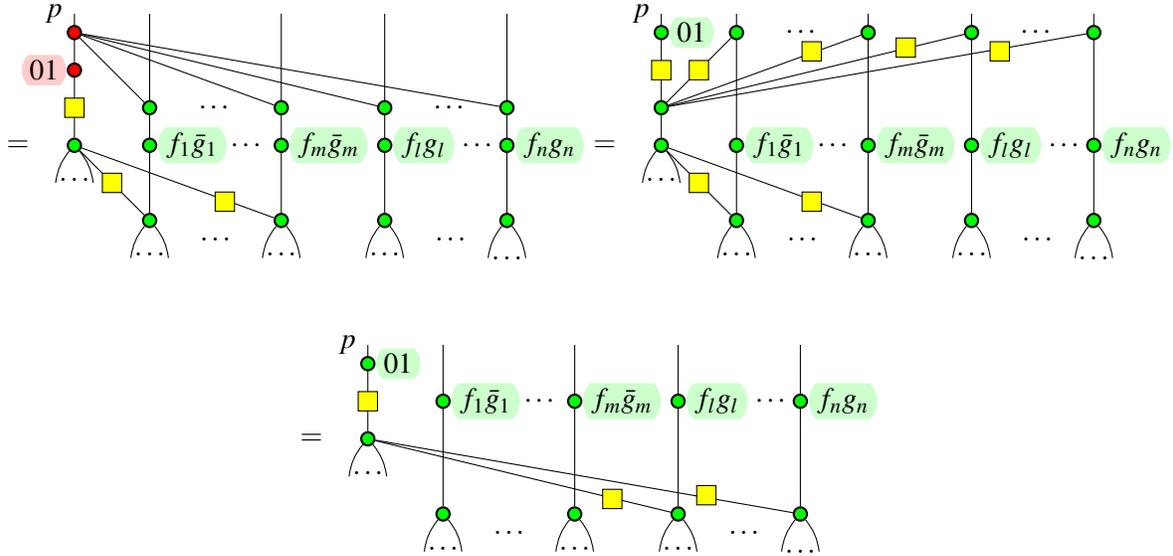
 \begin{center}
  \input{tikz_files/01_on_p_1.tikz}\input{tikz_files/01_on_p_2.tikz}
 \end{center}
 \begin{center}
  \input{tikz_files/01_on_p_3.tikz}\input{tikz_files/01_on_p_4.tikz}
 \end{center}
 \begin{center}
  \input{tikz_files/01_on_p_5.tikz}
 \end{center}
 Now if $N_1=N_2$, $p$ has no more neighbours and is in the state \state{rn,label={[rphase]right:$01$}}.
 This is not the same as the state $p$ has in diagram 1, so the diagrams are not equal.
 Else, after the application of $U$, $p$ still has some neighbours in diagram 2.
 Local complementations do not change this fact.
 Thus the two diagrams cannot be equal.
 The case $d=1,e=0$ is entirely analogous, except that there is a fixpoint operation in addition to the local complementation at the beginning.

 If $d=e=0$, there are two sub-cases.
 First, suppose there exists $v\in N_2$ such that $v\notin N_1$.
 Apply a local complementation about this $v$.
 This operation changes the vertex operator on $p$ to \phase{gn,label={[gphase]right:$01$}}.
 It also changes the edges incident on $p$, but the important thing is that $p$ will still have at least one neighbour.
 Thus we can proceed as in the case $d=0,e=1$.

 Secondly, suppose there is no $v\in N_2$ which is not in $N_1$.
 Since $N_2\neq\emptyset$ ($N_2=\emptyset$ corresponds to the case ``$p$ has no neighbours in $D_2$'', which was considered above), we must then be able to find $v\in N_1\cap N_2$.
 The diagram looks as follows, where now $m>0$ (again, we are ignoring edges that do not involve $p$):
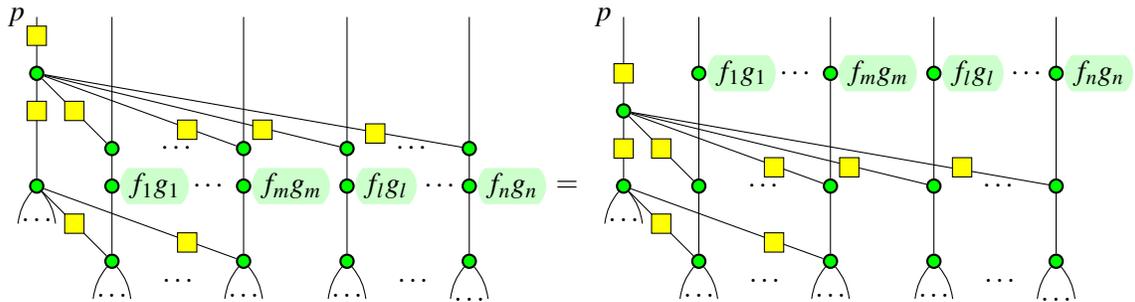
 \begin{center}
  \input{tikz_files/00_on_p_1.tikz}\input{tikz_files/00_on_p_2.tikz}
 \end{center}
 To show that the two diagrams are unequal it suffices to show that in diagram 2 the state of $p$ either factors out, but is not \state{rn} or \state{rn,label={[rphase]right:$11$}}, or that it remains entangled with other toy bits.
 We are thus justified in ignoring large portions of the above diagram to focus only on $p$, $v$ and the edge between the two.
 In particular, we will ignore for the moment the edges between $p$ and toy bits other than $v$, as well as the last \HadSpek\ on $p$.
 Then:
 \begin{center}
  \input{tikz_files/p20_1a.tikz}
 \end{center}
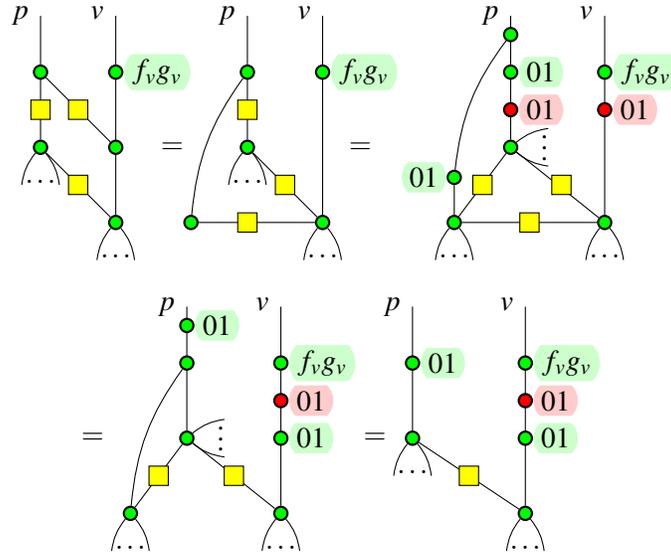
 \begin{center}
  \input{tikz_files/p20_1b.tikz}
 \end{center}
 where for the second equality we have applied a local complementation to $v$ and used the Euler decomposition, the third equality follows by a local complementation on $p$, and the last one comes from the merging of $p$ with the green node in the bottom left.
 Note that, in the end, $p$ and $v$ are still connected by an edge.
 None of the operations we ignored in picking out this part of the diagram will change that.
 Thus, as before, the state of $p$ cannot be the same as in diagram 1.
 The two diagrams are unequal.

 The case $d=e=1$ is analogous to $d=e=0$, except in either sub-case we start with a fixpoint operation on the chosen $v$.

 We have thus shown that a simplified pair of rGS-LO diagrams are not equal if there are any unpaired red nodes.
\end{proof}

\end{document}